\documentclass{article}%
\usepackage{amsmath}
\usepackage{amsfonts}
\usepackage{amssymb}
\usepackage{graphicx}%
\providecommand{\U}[1]{\protect\rule{.1in}{.1in}}
\newtheorem{theorem}{Theorem}

\newtheorem{definition}[theorem]{Definition}

\newtheorem{lemma}[theorem]{Lemma}

\newtheorem{remark}[theorem]{Remark}

\newenvironment{proof}[1][Proof]{\noindent\textbf{#1.} }{\ \rule{0.5em}{0.5em}}
\begin{document}

\title{On single-copy maximization of measured $f$-divergence between a given pair of
quantum states}
\author{Keiji Matsumoto\\National Institute of Informatics, \\2-1-2, Hitotsubashi, Chiyoda-ku, Tokyo 101-8430\\keiji@nii.ac.jp}
\maketitle

\begin{abstract}

\end{abstract}

This paper deals with (single-copy) maximization of classical $f$-divergence,
which is a generalization of Kullback-Leibler divergence and Renyi-type
relative entropy, of the distributions of measurement outputs of a given pair
of quantum states. So far, there had been significant development concerning
asymptotic maximization, or maximizing rate of the quantity when collective
measurements are performed on large number of copies of states. On the other
hand, however, relatively little has been done about single-copy maximization,
and the question is solved only for very restricted examples of $f$. The
pourpose of the present paper is to push forward the research by investigating
the properties of the maximized $f$-divergence, and rewriting the maximization
problem to more tractable form. The consequences of these efforts include an
expression of the maximized quantity by "non-commutative Radon-Nikodym
derivative", and closed formulas of the quantity in some special cases, e.g.,
when the first argument is a pure state,or $f$ is in some specific form.

\section{Introduction and summary of results}

This paper deals with (single-copy) maximization of classical $f$-divergence
between the distributions of a measurement outputs of a given pair of quantum
states. $f$-divergence $D_{f}$ between the probability density functions
$p_{1}$ and $p_{2}$ over a discrete set is defined as
\[
D_{f}\left(  p_{1}||p_{2}\right)  :=\sum_{x}p_{2}\left(  x\right)  f\left(
\frac{p_{1}\left(  x\right)  }{p_{2}\left(  x\right)  }\right)  ,
\]
if $p_{2}\left(  x\right)  >0$ for all $x$. (The definition for the general
case is given later.)

This problem is import since $D_{f}$ has good operational meanings. For
example, $f_{\mathrm{KL}}\left(  \lambda\right)  :=\lambda\ln\lambda$
correspond to Kullback-Leibler divergence. Also,
\begin{equation}
f_{\alpha}\left(  \lambda\right)  :=\mathrm{sign}\left(  \left(
\alpha-1\right)  \alpha\right)  \cdot\lambda^{\alpha}.\label{Renyi}%
\end{equation}
is essentially Rnyi-type relative entropy. (Here $\mathrm{sign}\left(
c\right)  $ is either $1$ or $-1$ depending on $c$ is non-negative or
negative, $\mathrm{sign}\left(  c\right)  =c/\left\vert c\right\vert $ .)They
play key role in the theory of large deviation, and thus extensively used in
asymptotic analysis of error probability of decoding, hypothesis test, and so on.

Other $f$-divergences than these have at least one operational meaning. If $f$
is a proper lower semicontinuous convex function whose domain contains
positive half-line, $D_{f}\left(  p_{1}||p_{2}\right)  $ is the optimal gain
of a certain Bayes decision problem. In other words, for each $f$ satisfying
above mentioned conditions, there is a pair of real valued functions $w_{1}$
and $w_{2}$ on decision space representing a gain of decision $d$, with
\begin{equation}
D_{f}\left(  p_{1}||p_{2}\right)  =\sup_{d\left(  \cdot\right)  }\sum
_{x}\left(  w_{1}\left(  d\left(  x\right)  \right)  p_{1}\left(  x\right)
+w_{2}\left(  d\left(  x\right)  \right)  p_{2}\left(  x\right)  \right)
.\label{sup-w-0}%
\end{equation}
Conversely, for each $\vec{w}\left(  \cdot\right)  =\left(  w_{1}\left(
\cdot\right)  ,w_{2}\left(  \cdot\right)  \right)  $, there is a proper lower
semicontinuous convex function $f$ satisfying the above identity. Also, by
(\ref{sup-w-0}) and the celebrated randomization criterion \cite{Strasser}%
\cite{Torgersen}, there is a Markov map which sends $\left(  p,q\right)  $ to
$\left(  p^{\prime},q^{\prime}\right)  $ if and only if $D_{f}\left(
p||q\right)  \geq D_{f}\left(  p^{\prime}||q^{\prime}\right)  $ holds for any
convex function $f$ .

In quantum information and quantum statistics, finding the optimal measured
$f$-divergence $D_{f}^{\min}$ is of interest:
\[
D_{f}^{\min}\left(  \rho_{1}||\rho_{2}\right)  :=\sup_{M\text{:POVM}}%
D_{f}\left(  P_{\rho_{1}}^{M}||P_{\rho_{2}}^{M}\right)  ,
\]
where POVM is the short for positive operator valued measures, and
$P_{\rho_{\theta}}^{M}$ is the probability distribution obtained by
application of the measurement $M$ to the state $\rho_{\theta}$ ($\theta
\in\left\{  1,2\right\}  $). (Underlying Hilbert space is finite dimensional
throughout the paper.) However, calculation of $D_{f}^{\min}\left(  \rho
_{1}||\rho_{2}\right)  $, being maximization of non-linear functional of POVM,
is far from tractable. Thus, so far, most of the works are devoted to
asymptotic analysis: for example, in their celebrated paper \cite{HiaiPetz},
Hiai and Petz showed that the asymptotic limit $\lim_{n\rightarrow\infty}%
\frac{1}{n}D_{f_{\mathrm{KL}}}^{\min}\left(  \rho_{1}^{\otimes n}||\rho
_{2}^{\otimes n}\right)  $ equals Umegaki-von Neumann type relative entropy.
In \cite{MosonyiOgawa}\cite{HayashiTomamichel}, they computed similar quantity
for $f_{\alpha}$.

However, if we turn to single-copy optimization, the problem had been so far
solved only for $f_{1/2}$, and $f\left(  \lambda\right)  =\left\vert
1-\lambda\right\vert $, which corresponds to fidelity and statistical
distance, respectively. Single-copy maximization is aso important, firstly
because collective measurements are technologically highly demanding, and
secondly because comparison between signle-conpy maximum and asymptotic one
signifies the effectivenss of collective measurements. Also, detailed
knowledge of single copy maximum may give deeper understanding of asymptotic theory.

A purpose of the present paper is to advance the study of exact (not
asymptotic) optimization, by rewriting the maximization problem to more
tractable form, and giving closed formulas of the quantity in some special
cases. Our main mathematical tool is convex analysis, and we exploit the
observation (\ref{sup-w-0}), or that $f$-divergence is supremum of linear
functional. 

One of the two main results in this direction is Theorem\thinspace
\ref{th:max-T}: if the convex conjugate $f^{\ast}$ of $f$ is operator convex,
$D_{f}^{\min}\left(  \rho_{1}||\rho_{2}\right)  $ is written as the supremum
of a concave function of a Hermitian operator \textit{without} extending the
underlying Hilbert space, namely,%
\[
D_{f}^{\min}\left(  \rho_{1}||\rho_{2}\right)  =\sup\left\{  \mathrm{tr}%
\,\rho_{1}T-\mathrm{tr}\,\rho_{2}f^{\ast}\left(  T\right)  ;\mathrm{spec}%
\,T\subset\mathrm{dom}\,f^{\ast}\right\}  .
\]
With some additional assumptions on $f$, \ this further can be rewritten as
\[
D_{f}^{\min}\left(  \rho_{1}||\rho_{2}\right)  =\mathrm{tr}\,\rho_{2}f\left(
f^{\ast\prime}\left(  T_{0}\right)  \right)  ,
\]
where, denoting the Frechet derivative of $f^{\ast}$ by $\mathrm{D}f^{\ast}$,
$T_{0}$ is a solution to the matrix equation
\[
\rho_{1}=\mathrm{D}f^{\ast}\left(  T_{0}\right)  \left(  \rho_{2}\right)  .
\]
This equation is straight forwardly solved if $f=$ $f_{2}$ and $f_{-1}$,
resulting in closed formla of the quantity in each case.

The other main result is Theorem \ref{th:with-kernel}, which simplifies the
optimization when $\rho_{1}$ has non-trivial kernel: if $f^{\ast}$ is operator
convex, $\mathrm{dom}\,f^{\ast}$ is unbounded from below, one can reduce the
optimization problem to the one in $\mathrm{supp\,}\rho_{1}$, namely,
\[
D_{f}^{\min}\left(  \rho_{1}||\rho_{2}\right)  =D_{f}^{\min}\left(  \rho
_{1}||\pi_{\rho_{1}}\rho_{2}\pi_{\rho_{1}}\right)  +f\left(  0\right)  \left(
1-\mathrm{tr}\,\rho_{2}\pi_{\rho_{1}}\right)  ,
\]
where $\pi_{\rho_{1}}$ is the projection onto $\mathrm{supp\,}\rho_{1}$. If
$\rho_{1}\,$\ is a pure state, $\rho_{1}=\left\vert \varphi_{1}\right\rangle
\left\langle \varphi_{1}\right\vert $, \ this is gives closed formula,
\[
D_{f}^{\min}\left(  \left\vert \varphi_{1}\right\rangle \left\langle
\varphi_{1}\right\vert \,||\rho_{2}\right)  =\left\langle \varphi
_{1}\right\vert \rho_{2}\left\vert \varphi_{1}\right\rangle f\left(  \frac
{1}{\left\langle \varphi_{1}\right\vert \rho_{2}\left\vert \varphi
_{1}\right\rangle }\right)  +f\left(  0\right)  \left(  1-\left\langle
\varphi_{1}\right\vert \rho_{2}\left\vert \varphi_{1}\right\rangle \right)  .
\]

Using above results, we analyze $D_{f}^{\min}$ of infinitesimally close two
staes, up to the second order of the distance between them:
\[
\lim_{\eta^{\prime}\rightarrow\eta}\frac{1}{\left(  \eta^{\prime}-\eta\right)
^{2}}D_{f}^{\min}\left(  \rho_{\eta}||\rho_{\eta^{\prime}}\right)  ,
\]
where $\left\{  \rho_{\eta}\right\}  _{\eta\in\mathbb{R}}$ is a family of
parameterized states and $f$ has good properties. It have been a folklore that
this limit equals the constant multiple of SLD (symmetric logarithmic
derivative) Fisher information, which plays an important role in the
asymptotic theory of statistics. We prove the folklore in the case where
$\rho_{\eta}$ 's rank is either full or one, and disprove it in other cases,
giving the alternative correct formula (\ref{D=JS-2}).

\section{Classical $f$-divergence}

\label{sec:classical}

In this section, we summarize known facts about classical $f$-divergence and
convex functions. As in \cite{Rockafellar}, we suppose that $f$ is a map from
$\mathbb{R}^{n}$ $\ $to $\mathbb{R\cup}\left\{  \pm\infty\right\}  $. Instead
of saying that $f$ is not defined on a certain set, we say that $f\left(
\lambda\right)  =\infty$ on that set. 

\begin{definition}
The effective domain of $f$, denoted by $\mathrm{dom}\,f$ ,\ is the set of all
$\lambda$'s with $f\left(  \lambda\right)  \mathbb{<\infty}$. $f$ is a
\textit{convex function} if and only if its \textit{epigraph}, or the set
\[
\mathrm{epi}f:=\left\{  \left(  \lambda_{1},\lambda_{2}\right)  ;\lambda
_{2}\geq f\left(  \lambda_{1}\right)  \right\}
\]
is convex. A function $f$ is proper if and only if $f$ is nowhere $-\infty$
and not $\infty$ everywhere, and is lower semi-continuous if and only if the
set $\left\{  \lambda_{1};\lambda_{2}\geq f\left(  \lambda_{1}\right)
\right\}  $ is closed for any $\lambda_{2}$. \ Given a convex functionion $f$,
its \textit{lower semi-continuous hull} is a greatest lower semi-continuous
function (not necessarily finite) majorized by $f$. 
\end{definition}

\begin{remark}
An improper convex function, being necessarily infinite except perhaps at
relative boundary points of its effective doimain (Theorem\thinspace7.2 of
\cite{Rockafellar}), rarely appears in application.
\end{remark}

\begin{lemma}
(Theorem\thinspace7.1 of \cite{Rockafellar}) A proper convex function $f$ is
lower semi-continuous if and only if its epigraph is closed. Given a convex
functionion $f$, its lower semi-continuous hull always exists, and coincide
with $f$ except perhaps at the relative boundary points of its effective
domain.\ The epigraph of the lower semi-continuous hull is the closure of the
epigraph of $f$.
\end{lemma}

\begin{lemma}
\label{lem:p-sup}(Corollary 13.5.1 of \cite{Rockafellar}) If $f$ is proper and
lower semi-continuous and convex, it is the pointwise supremum of linear functions.
\end{lemma}

From here, unless otherwise mentioned, functions denoted by $f$, $f_{\alpha}$
etc are a proper lower semi-continuous convex function on $\mathbb{R}$, and
its effective domain contains $\left(  0,\infty\right)  $. For a given $f$,
$f$-divergence $D_{f}\left(  P_{1}||P_{2}\right)  $ between positive finite
measures $P_{1}$ and $P_{2}$ is defined as follows. 

\begin{definition}
Let $p_{1}$ and $p_{2}$ be density function of $P_{1}$ and $P_{2}$ with
respect to a measure $\mu$ which dominates them. (Such a measure $\mu$ always
exists. For example, one may choose $\mu=P_{1}+P_{2}$.)  When $p_{1}$ and
$p_{2}$ have the common support (e.g. the set where they are positive),
\[
D_{f}\left(  P_{1}||P_{2}\right)  :=\int_{\mathrm{supp}\,p_{2}}p_{2}\left(
x\right)  f\left(  \frac{p_{1}\left(  x\right)  }{p_{2}\left(  x\right)
}\right)  \mathrm{d}\mu\left(  x\right)  .
\]
(This definition seemingly depends on the choice of $\mu$, but in fact it does
not.) When their supports are not identical, we extend $D_{f}$ so that the
function $\left(  P_{1},P_{2}\right)  \rightarrow D_{f}\left(  P_{1}%
||P_{2}\right)  $ is the pointwise supremum of linear functions:
\[
D_{f}\left(  P_{1}||P_{2}\right)  :=\int g\left(  p_{1}\left(  x\right)
,p_{2}\left(  x\right)  \right)  \mathrm{d}\mu\left(  x\right)  ,
\]
where $g\left(  \lambda_{1},\lambda_{2}\right)  $ is the lower semi-continuous
hull of $\lambda_{2}f\left(  \frac{\lambda_{1}}{\lambda_{2}}\right)  $, or
more explicitly ( see p.\thinspace35 and p.67 of \cite{Rockafellar} ),
\[
g\left(  \lambda_{1},\lambda_{2}\right)  :=\left\{
\begin{array}
[c]{cc}%
\lambda_{2}f\left(  \frac{\lambda_{1}}{\lambda_{2}}\right)  , & \text{if
}\lambda_{1}\in\mathrm{dom}\,f,\lambda_{2}>0\\
\lim_{\lambda_{2}\downarrow0}\lambda_{2}f\left(  \frac{\lambda_{1}}%
{\lambda_{2}}\right)  , & \text{if }\lambda_{1}\,\in\mathrm{dom}%
\,f,\,\lambda_{2}=0,\\
0, & \text{if }\lambda_{1}=\lambda_{2}=0,\\
\infty, & \text{if }\lambda_{1}\not \in \mathrm{dom}\,f\text{\ \ or }%
\lambda_{2}<0.
\end{array}
\right.
\]

\end{definition}

That the function $\left(  P_{1},P_{2}\right)  \rightarrow D_{f}\left(
P_{1}||P_{2}\right)  $ defined above is the pointwise supremum of linear
functions is proved as follows. Since $g$ is proper and lower semi-continuous
and convex, by Lemma\thinspace\ref{lem:p-sup}, it is the pointwise supremum of
linear functions,%
\begin{equation}
g\left(  \lambda_{1},\lambda_{2}\right)  =\sup_{\vec{w}\in\mathcal{W}}%
\sum_{\theta\in\left\{  1,2\right\}  }w_{\theta}\lambda_{\theta}%
,\label{g=supW}%
\end{equation}
where \
\[
\mathcal{W}_{f}:\mathcal{=}\left\{  \vec{w}=\left(  w_{\theta}\right)
_{\theta\in\left\{  1,2\right\}  };w_{1}\leq w_{1}^{\prime},\,w_{2}%
\leq-f^{\ast}\left(  w_{1}^{\prime}\right)  ,\exists w_{1}^{\prime}%
\in\mathrm{dom}\,f^{\ast}\right\}  ,
\]
and $f^{\ast}$ is convex conjugate of $f$,
\[
f^{\ast}\left(  t\right)  :=\sup_{\lambda\in\mathbb{R}}\left(  t\lambda
-f\left(  \lambda\right)  \right)  .
\]

Therefore, as will be explained below, \
\begin{align}
&  D_{f}\left(  P_{1}||P_{2}\right)  \nonumber\\
&  =\int\left(  \sup_{\vec{w}\in\mathcal{W}}\sum_{\theta\in\left\{
1,2\right\}  }w_{\theta}p_{\theta}\left(  x\right)  \right)  \mathrm{d}%
\mu\left(  x\right)  \nonumber\\
&  =\sup\left\{  \int\sum_{\theta\in\left\{  1,2\right\}  }w_{\theta}\left(
x\right)  p_{\theta}\left(  x\right)  \mathrm{d}\mu\left(  x\right)  ;\vec
{w}\left(  \cdot\right)  \text{: bounded, measurable, into }\mathcal{W}%
_{f}\text{ }\right\}  .\label{D=sup-w}%
\end{align}
Thus we have the assertion. Here, the second identity above holds since the
function
\[
x\rightarrow\sup_{\vec{w}\in\mathcal{W}}\sum_{\theta\in\left\{  1,2\right\}
}w_{\theta}p_{\theta}\left(  x\right)
\]
is meaureable, and any measureable function is aribtrarily approximated by a
simple function, which is bounded.

\begin{remark}
Since $g$ is in addition positively homogeneous, or
\[
\forall a\geq0,\,\,g\left(  a\lambda_{1},a\lambda_{2}\right)  =ag\left(
\lambda_{1},\lambda_{2}\right)  ,\,
\]
it is obvious that the value of $D_{f}$ does not depends on the choice of
$\mu$, despite its apparent dependency on $\mu$.
\end{remark}

\begin{remark}
(\ref{D=sup-w}) indicates (\ref{sup-w-0}).To see this, use $\mathcal{W}_{f}$
as a decision space, and let the the gain of the decision $\vec{w}$ be
$w_{\theta}$ when true probability distribution is $P_{\theta}$.
\end{remark}

\begin{remark}
If $p_{2}\left(  x\right)  =0$ and $p_{1}\left(  x\right)  >0$,
\begin{align*}
g\left(  p_{1}\left(  x\right)  ,p_{2}\left(  x\right)  \right)   &
=\lim_{\lambda_{2}\downarrow0}\lambda_{2}f\left(  \frac{p_{1}\left(  x\right)
}{\lambda_{2}}\right)  ,\\
&  =p_{1}\left(  x\right)  \lim_{\lambda_{2}^{\prime}\downarrow0}\lambda
_{2}^{\prime}f\left(  \frac{1}{\lambda_{2}^{\prime}}\right)  ,
\end{align*}
and if $p_{1}\left(  x\right)  =0$ and $p_{2}\left(  x\right)  =0$, $g\left(
p_{1}\left(  x\right)  ,p_{2}\left(  x\right)  \right)  =0$. Thus,
\[
D_{f}\left(  P_{1}||P_{2}\right)  =\int_{\mathrm{supp}\,p_{2}}p_{2}\left(
x\right)  f\left(  \frac{p_{1}\left(  x\right)  }{p_{2}\left(  x\right)
}\right)  \mathrm{d}\mu\left(  x\right)  +P_{1}\left(  \left\{  x;p_{2}\left(
x\right)  =0,p_{1}\left(  x\right)  \neq0\right\}  \right)  \lim_{\lambda
_{2}\downarrow0}\lambda_{2}f\left(  \frac{1}{\lambda_{2}}\right)  .
\]

\end{remark}

The correspondence between $D_{f}$ and $\mathcal{W}_{f}$ is one-to-one, but
the one between $\mathcal{W}_{f}$ and $f$ is not, since only monotone
increasing part of $f^{\ast}$ is relavant in the definition of $\mathcal{W}%
_{f}$, or equivalently, the values of $f$ in the negative half line does not
affect the value of $D_{f}$.

\begin{definition}
A convex function $f$ on $\mathbb{R}$ is \textit{canonical} if $f^{\ast}$ is
strictly monotone increasing on $\mathrm{dom}\,f^{\ast}$. If $f^{\ast}$ is not
canonical, we define its "canonicalization" $\ f_{0}$ by
\[
f_{0}^{\ast}\left(  t\right)  :=\left\{
\begin{array}
[c]{cc}%
f^{\ast}\left(  t\right)  , & \text{if }\,\forall t^{\prime}>t,\,\,f^{\ast
}\left(  t^{\prime}\right)  >f^{\ast}\left(  t\right)  ,\,\\
+\infty & \text{otherwise.}%
\end{array}
\right.
\]

\end{definition}

Recall that $f$ used in the definition of $f$-divergence is not always
$+\infty$ on the positive half line. Therefore, $\mathrm{dom}\,f_{0}^{\ast}$
is not empty.

Denoting by $f_{+}^{\prime}$ the right derivative of $f$, in fact we have
\begin{equation}
f_{0}\left(  \lambda\right)  =\left\{
\begin{array}
[c]{cc}%
f\left(  \lambda\right)  , & \text{if }\lambda\geq0,\\
f_{+}^{\prime}\left(  0\right)  t+f\left(  0\right)  , & \text{if }\lambda<0.
\end{array}
\right.  \label{f_0-1}%
\end{equation}
if $f_{+}^{\prime}\left(  0\right)  $ is finite, and
\begin{equation}
f_{0}\left(  \lambda\right)  =\left\{
\begin{array}
[c]{cc}%
f\left(  \lambda\right)  , & \text{if }\lambda\geq0,\\
+\infty, & \text{if }\lambda<0,
\end{array}
\right.  \label{f_0-2}%
\end{equation}
otherwise. In addition, $t_{0}$ is the largest lower bound to $\mathrm{dom}%
\,f_{0}^{\ast}$
\begin{equation}
t_{0}=f_{+}^{\prime}\left(  0\right)  ,\label{t0=f'}%
\end{equation}
where
\[
t_{0}:=\inf\left\{  t;t\in\mathrm{dom}\,f_{0}^{\ast}\right\}  .
\]
(If $\mathrm{dom}\,f_{0}^{\ast}$ is not bounded from below, $t_{0}:=-\infty$.)
By definition of $f_{0}^{\ast}$,
\begin{equation}
\inf_{t}\,f\left(  t\right)  =\inf_{t}\,f^{\ast}\left(  t\right)  =f\left(
t_{0}\right)  .\label{t0=inf}%
\end{equation}
The proof of (\ref{f_0-1}), (\ref{f_0-2}) and (\ref{t0=f'}) are given in
Appendix\thinspace\ref{sec:proof-f_0}.

For example, both of
\begin{align*}
f_{\mathrm{TV}}\left(  \lambda\right)   &  :=\left\vert 1-\lambda\right\vert
,\\
\widetilde{f}_{\mathrm{TV}}\left(  \lambda\right)   &  :=\left\{
\begin{array}
[c]{cc}%
\left\vert 1-\lambda\right\vert , & \lambda\geq0,\\
\infty, & \lambda<0,
\end{array}
\right.
\end{align*}
correspond to the total variation distance,
\[
D_{f_{\mathrm{TV}}}\left(  P||Q\right)  =D_{\widetilde{f}_{\mathrm{TV}}%
}\left(  P||Q\right)  =\left\Vert P-Q\right\Vert _{1}.
\]
The former is canonical but the latter is not,%
\begin{align*}
f_{\mathrm{TV}}^{\ast}\left(  t\right)   &  =\left\{
\begin{array}
[c]{cc}%
\infty, & t<-1,\\
t, & -1\leq t\leq1,\\
\infty, & t>1,
\end{array}
\right.  ,\\
\widetilde{f}_{\mathrm{TV}}^{\ast}\left(  t\right)   &  =\left\{
\begin{array}
[c]{cc}%
-1, & t<-1,\\
t, & -1\leq t\leq1,\\
\infty, & t>1,
\end{array}
\right.  .
\end{align*}

\begin{definition}
Define
\end{definition}%

\[
\hat{f}\left(  \lambda\right)  :=g\left(  1,\lambda\right)  ,\,\lambda\geq0,
\]
and canonically extend it to the negative half line. Also, define $\hat{g}$ by
in a parallel manner as the definition of $g$, replacing $f$ by $\hat{f}$. 

By definition,
\begin{equation}
\hat{g}\left(  \lambda_{1},\lambda_{2}\right)  =g\left(  \lambda_{2}%
,\lambda_{1}\right)  ,\,\label{g=g}%
\end{equation}
holds for all $\lambda_{1}\geq0$: if $\lambda_{1}\in\mathrm{dom}\,\hat{f}$ and
$\lambda_{2}>0$, the relation is checked by easy computation. If $\lambda
_{1}=\lambda_{2}=0$, the both ends of (\ref{g=g}) is $0$. Finally, if
$\ \lambda_{1}=0\notin\mathrm{dom}\,\hat{f}$ and $\lambda_{2}>0$,%
\[
\hat{g}\left(  0,\lambda_{2}\right)  =\infty=\hat{f}\left(  0\right)
=\lambda_{2}\hat{f}\left(  0\right)  =\lambda_{2}g\left(  1,0\right)
=g\left(  \lambda_{2},0\right)  ,
\]
and the relation is checked.(\ref{g=g}) means
\[
D_{f}\left(  P_{1}||P_{2}\right)  =D_{\hat{f}}\left(  P_{2}||P_{1}\right)  .
\]

\section{Expressions of \ \ $D_{f}^{\min}$ (I)}

\label{sec:expressions-1}

Consider density operators $\left\{  \rho_{\theta}\right\}  _{\theta
\in\left\{  1,2\right\}  }$ over finite dimensional Hilbert space
$\mathcal{H}$ ( throughout the paper, Hilbert spaces are always finite
dimensional) .

\begin{definition}
\begin{equation}
D_{f}^{\min}\left(  \rho_{1}||\rho_{2}\right)  :=\sup_{M\text{:POVM}}%
D_{f}\left(  P_{\rho_{1}}^{M}||P_{\rho_{2}}^{M}\right)  ,\label{Dfmin}%
\end{equation}
where POVM is the short for positive operator valued measures, and
$P_{\rho_{\theta}}^{M}$ is the probability distribution resulting from
application of the measurement $M$ to the state $\rho_{\theta}$ ($\theta
\in\left\{  1,2\right\}  $).
\end{definition}

The notation $D_{f}^{\min}$ comes from the following fact. If $D_{f}^{Q}$ is a
real valued function of $\left\{  \rho_{\theta}\right\}  _{\theta\in\left\{
1,2\right\}  }$ which coincide with $D_{f}$ on any commutative subalgebra and
is monotone non-increasing by application of completely positive trace
preserving (CPTP) maps, then%

\begin{equation}
D_{f}^{\min}\left(  \rho_{1}||\rho_{2}\right)  \leq D_{f}^{Q}\left(  \rho
_{1}||\rho_{2}\right)  .\label{Dfmin<DfQ}%
\end{equation}
since, by definition of $D_{f}^{\min}$, for any $\varepsilon>0$, there is a
measurement $M$ such that
\begin{align*}
D_{f}^{\min}\left(  \rho_{1}||\rho_{2}\right)  -\varepsilon & \leq
D_{f}\left(  P_{\rho_{1}}^{M}||P_{\rho_{2}}^{M}\right)  \\
& =D_{f}^{Q}\left(  P_{\rho_{1}}^{M}||P_{\rho_{2}}^{M}\right)  \\
& \leq D_{f}^{Q}\left(  \rho_{1}||\rho_{2}\right)  .
\end{align*}

For the sake of notational simplicity, we extend $D_{f}$ and $D_{f}^{\min}$ to
all the positive finite measures and all the positive operators.

Denote also by $p_{\rho_{\theta}}^{M}$ the density of $P_{\rho_{1}}^{M}$ with
respect to an underlying measure $\mu^{M}$, which may be taken as $P_{\rho
_{1}}^{M}+P_{\rho_{2}}^{M}$. Using (\ref{D=sup-w}),%
\begin{align*}
D_{f}^{\min}\left(  \rho_{1}||\rho_{2}\right)   &  =\sup_{\vec{w}\left(
\cdot\right)  ,M}\int\sum_{\theta\in\left\{  1,2\right\}  }w_{\theta}\left(
x\right)  p_{\rho_{\theta}}^{M}\left(  x\right)  \mathrm{d}\mu^{M}\left(
x\right)  \\
&  =\text{ }\sup_{\vec{w}\left(  \cdot\right)  ,M}\int\sum_{\theta\in\left\{
1,2\right\}  }w_{\theta}\left(  x\right)  \mathrm{tr}\,\rho_{\theta}M\left(
\mathrm{d}x\right)
\end{align*}
where $\vec{w}\left(  \cdot\right)  $ moves all over the simple functions into
$\mathcal{W}_{f}$. \ By defining a new POVM $M^{\prime}\left(  B\right)
:=M\left(  \vec{w}^{-1}\left(  B\right)  \right)  $,  \
\[
D_{f}^{\min}\left(  \rho_{1}||\rho_{2}\right)  =\sup_{M}\int_{\mathcal{W}_{f}%
}\sum_{\theta\in\left\{  1,2\right\}  }w_{\theta}\mathrm{tr}\,\rho_{\theta
}M\left(  \mathrm{d}\vec{w}\right)  .
\]

Since the functional
\[
M\rightarrow\int_{\mathcal{W}_{f}}\sum_{\theta\in\left\{  1,2\right\}
}w_{\theta}\mathrm{tr}\,\rho_{\theta}M\left(  \mathrm{d}\vec{w}\right)
\]
is affine, by Caratheodory's theorem, the support of $M$ can be reduced to a
finite set without changing the supremum. (The cardinality of $\mathrm{supp}%
\,M$ may be less than or equal to $\left(  \dim\mathcal{H}\right)  ^{2}+2$.)
Thus taking $\mu$ as the counting measure, we have
\begin{equation}
D_{f}^{\min}\left(  \rho_{1}||\rho_{2}\right)  =\sup_{M}\sum_{\vec{w}%
\in\mathrm{supp}\,M}\sum_{\theta\in\left\{  1,2\right\}  }w_{\theta
}\,\mathrm{tr}\,\rho_{\theta}M_{\vec{w}}.\label{D=sup-w-finite}%
\end{equation}
In taking supremum, without loss of generality, we may restrict the support of
$M$ to the boundary of $\mathcal{W}_{f}$ . Thus,
\begin{align}
D_{f}^{\min}\left(  \rho_{1}||\rho_{2}\right)   &  =\sup_{M}\sum
_{t\in\mathrm{supp}\,M}\left\{  t\,\mathrm{tr}\,\rho_{1}M_{t}-f^{\ast}\left(
t\right)  \mathrm{tr}\,\rho_{2}M_{t}\right\}  ,\label{D=sup-t}\\
&  =\sup_{M}\sum_{s\in\mathrm{supp}\,M}\left\{  -\hat{f}^{\ast}\left(
s\right)  \,\mathrm{tr}\,\rho_{1}M_{s}+s\mathrm{tr}\,\rho_{2}M_{s}\right\}
,\label{D=sup-s}%
\end{align}
where $\mathrm{supp}\,M$ in (\ref{D=sup-t}) and (\ref{D=sup-s}) is a finite
subset of $\mathrm{dom}\,f^{\ast}$ and $\mathrm{dom}\,\hat{f}^{\ast}$, respectively.

\section{Necessary and sufficient conditions for $D_{f}^{\min}<\infty$}

\label{sec:Df<infty}

In this section we determine the case where $D_{f}^{\min}$ stays finite. Note
this result also give, due to (\ref{Dfmin<DfQ}), the necessary (and in fact ,
sufficient) condition that all the quantum versions of $f$-divergence becomes finite.

\begin{definition}
\label{def:b-t}Define \ $b_{\ast}$ (,$b_{\ast}^{\prime}$, \ resp.) as the
smallest number with $\rho_{1}-b_{\ast}\rho_{2}\leq0$ (, the largest number
with $\rho_{1}-b_{\ast}^{\prime}\rho_{2}\geq0$). When such $b_{\ast}$ does not
exist, we let $b_{\ast}:=\infty$. Let $t_{\ast}$ ($t_{\ast}^{\prime}$, resp.)
be the largest (smallest, resp.) number such that $b_{\ast}$ ($b_{\ast
}^{\prime}$, resp.) is a subgradient of $f^{\ast}$ at $t_{\ast}$ ($t_{\ast
}^{\prime}$, resp.).  If $b_{\ast}=\infty$ ($b_{\ast}^{\prime}=0$, resp.), we
define $t_{\ast}$ as the smallest upper bound to $\mathrm{dom}\,f^{\ast}$
($t_{\ast}^{\prime}$ as the largest lower bound to $\mathrm{dom}\,f^{\ast}$,
resp.). 
\end{definition}

If $b_{\ast}<\infty$ ($b_{\ast}>0$, resp.), such $t_{\ast}$ ($t_{\ast}%
^{\prime}$, resp.) exists, since the effective domain of $f^{\ast\ast}=f$ , or
equivalently the range of the subgradient, contains $\left(  0,\infty\right)
$ ($\left(  0,\infty\right)  \subset\mathrm{dom}\,f$ is supposed throughout
the paper, as stated in Section\thinspace\ref{sec:classical}).\ By definition,
\ $b_{\ast}\geq b_{\ast}^{\prime}$ and $t_{\ast}\geq t_{\ast}^{\prime}$.

Observe, by definition of the subgradient and convexity of $f^{\ast}$,
\[
t\leq t_{\ast}^{\prime}\Rightarrow f^{\ast}\left(  t_{\ast}^{\prime}\right)
-f^{\ast}\left(  t\right)  \leq b_{\ast}^{\prime}\left(  t_{\ast}^{\prime
}-t\right)  ,
\]
and
\begin{align*}
t  & \geq t_{\ast}\Rightarrow f^{\ast}\left(  t\right)  -f^{\ast}\left(
t_{\ast}\right)  \geq b_{\ast}\left(  t-t_{\ast}\right)  \\
& \Leftrightarrow f^{\ast}\left(  t_{\ast}\right)  -f^{\ast}\left(  t\right)
\leq b_{\ast}\left(  t_{\ast}-t\right)  .
\end{align*}
Therefore, \ for all $t\leq t_{\ast}^{\prime}$%

\begin{align*}
& \left\{  t\,_{\ast}^{\prime}\rho_{1}\,-f^{\ast}\left(  t_{\ast}^{\prime
}\right)  \rho_{2}\right\}  -\left\{  t\rho_{1}-f^{\ast}\left(  t\right)
\rho_{2}\right\}  \\
& \geq\left(  t_{\ast}^{\prime}-t\right)  \left\{  \rho_{1}-b_{\ast}^{\prime
}\rho_{2}\right\}  \geq0,\text{ }%
\end{align*}
and for all $t\geq t_{\ast}$,w \
\begin{align*}
& \left\{  t\,_{\ast}\rho_{1}\,-f^{\ast}\left(  t_{\ast}\right)  \rho
_{2}\right\}  -\left\{  t\rho_{1}-f^{\ast}\left(  t\right)  \rho_{2}\right\}
\\
& \geq\left(  t_{\ast}-t\right)  \left\{  \rho_{1}-b_{\ast}\rho_{2}\right\}
\geq0.
\end{align*}
Therefore, concentration of the support of $M$ on $\left[  t_{\ast}^{\prime
},t_{\ast}\right]  $ does not decrease the sum%

\[
\sum_{t\in\mathrm{supp}\,M}\mathrm{tr}\left\{  t\,\rho_{1}-f^{\ast}\left(
t\right)  \rho_{2}\right\}  M_{t}.
\]
In other words, we can suppose
\begin{equation}
\mathrm{supp}\,M\subset\left[  t_{\ast}^{\prime},t_{\ast}\right]
.\label{suppM}%
\end{equation}

First suppose $\mathrm{supp}\,\rho_{1}=\mathrm{supp}\,\rho_{2}$. In this case,
$b_{\ast}<\infty$ and $b_{\ast}^{\prime}>0$. Thus, $t_{\ast}$ and $t_{\ast
}^{\prime}$ are finite. Therefore,
\[
D_{f}^{\min}\left(  \rho_{1}||\rho_{2}\right)  \leq t_{\ast}-f^{\ast}\left(
t_{\ast}^{\prime}\right)  <\infty.
\]

Next, we study the case where $\mathrm{supp}\rho_{1}\subset\mathrm{supp}%
\,\rho_{2}$ and $\ker\rho_{1}\neq\left\{  0\right\}  $. Let $\pi$ be the
projection onto $\ker\rho_{1}$. Then by (\ref{D=sup-w-finite}),
\begin{align*}
D_{f}^{\min}\left(  \rho_{1}||\rho_{2}\right)   &  \geq w_{1}\mathrm{tr}%
\,\rho_{1}\pi+w_{2}\mathrm{tr}\,\rho_{2}\pi+w_{1}^{\prime}\mathrm{tr}%
\,\rho_{1}\left(  \mathbf{1}-\pi\right)  +w_{2}^{\prime}\mathrm{tr}\,\rho
_{2}\left(  \mathbf{1}-\pi\right) \\
&  =w_{2}\mathrm{tr}\,\rho_{2}\pi+w_{1}^{\prime}\mathrm{tr}\,\rho_{1}\left(
\mathbf{1}-\pi\right)  +w_{2}^{\prime}\mathrm{tr}\,\rho_{2}\left(
\mathbf{1}-\pi\right)
\end{align*}
holds for any any $\vec{w}$, $\vec{w}^{\prime}\in$ $\mathcal{W}_{f}$.
Therefore, $D_{f}^{\min}\left(  \rho_{1}||\rho_{2}\right)  $ can be finite
only if $w_{2}$ stays finite, i.e.,
\begin{equation}
\mathcal{W}_{f}\subset\mathbb{R}\times(-\infty,a_{2}]. \label{Wf-bounded-1}%
\end{equation}

On the other hand, suppose the above inclusion is true. Since $\mathrm{supp}%
\rho_{1}\subset\mathrm{supp}\,\rho_{2}$, $b_{\ast}>0$ and $t_{\ast}$ is
finite. Therefore, by (\ref{suppM}) and (\ref{Wf-bounded-1}),
\[
D_{f}^{\min}\left(  \rho_{1}||\rho_{2}\right)  \leq t_{\ast}+a_{2}<\infty.
\]

Exchanging the role of $w_{1}$ and $w_{2}$, thus replacing $f^{\ast}$ and
(\ref{Wf-bounded-1}) by $\hat{f}^{\ast}$ and
\begin{equation}
\mathcal{W}_{f}\subset(-\infty,a_{1}]\times\mathbb{R}, \label{Wf-bounded-2}%
\end{equation}
respectively, the case where $\mathrm{supp}\rho_{2}\subset\mathrm{supp}%
\,\rho_{1}$ and $\ker\rho_{2}\neq\left\{  0\right\}  $ is almost analogously analyzed.

Finally, suppose $\mathrm{supp}\,\rho_{1}$ $\not \subset \mathrm{supp}%
\,\,\rho_{2}$ and $\mathrm{supp}\,\rho_{2}$ $\not \subset \mathrm{supp}%
\,\,\rho_{1}$. Then by the argument almost parallel to the one used to show
\[
D_{f}^{\min}<\infty\Rightarrow\text{(\ref{Wf-bounded-1}),}%
\]
$D_{f}^{\min}$ is finite only if there are finite numbers $a_{1}$ and $a_{2}$
with
\begin{equation}
\mathcal{W}_{f}\subset(-\infty,a_{1}]\times(-\infty,a_{2}].\label{Wf-bounded}%
\end{equation}
(This is the case if $\mathrm{dom}\,f^{\ast}$ is a finite interval.) On the
other hand, if this condition is true, by (\ref{D=sup-w}),
\[
D_{f}\left(  P_{1}||P_{2}\right)  \leq\max\left\{  a_{1},a_{2}\right\}
<\infty.
\]

To summarize:

\begin{theorem}
\label{th:Df-finite}Suppose $f$ is a proper lower semicontinuous convex
function with $\mathrm{dom}\,f\,\supset(0,\infty)$. Then, \ $D_{f}^{\min
}\left(  \rho_{1}||\rho_{2}\right)  <\infty$ \ holds if $\mathrm{supp}%
\,\rho_{1}$ $=\mathrm{supp}\,\,\rho_{2}$. Also, $D_{f}^{\min}\left(  \rho
_{1}||\rho_{2}\right)  <\infty$ is equivalent to (\ref{Wf-bounded}) if
$\mathrm{supp}\,\rho_{1}$ $\not \subset \mathrm{supp}\,\,\rho_{2}$ and
$\mathrm{supp}\,\rho_{2}$ $\not \subset \mathrm{supp}\,\,\rho_{1}$ hold, to
(\ref{Wf-bounded-1}) if $\mathrm{supp}\rho_{1}\subset\mathrm{supp}\,\rho_{2}$
and $\ker\rho_{1}\neq\left\{  0\right\}  $ hold, and to (\ref{Wf-bounded-2}%
)\ \ if $\mathrm{supp}\rho_{2}\subset\mathrm{supp}\,\rho_{1}$ and $\ker
\rho_{2}\neq\left\{  0\right\}  $ hold.
\end{theorem}

Note (\ref{Wf-bounded-1}) and (\ref{Wf-bounded-2}) is equivalent to
\[
f^{\ast}\left(  -\infty\right)  >-\infty
\]
and
\[
\mathrm{dom}\,f^{\ast}\subset(-\infty,a_{1}],
\]
respectively.

As proved in Appendix\thinspace\ref{sec:proof-th-finite}, one can in fact show:

\begin{theorem}
\label{th:DfQ-finite}Suppose $f$ is a proper lower semicontinuous convex
function with $\mathrm{dom}\,f\,\supset(0,\infty)$. Suppose also $D_{f}^{Q}$
is a real valued function of $\left\{  \rho_{\theta}\right\}  _{\theta
\in\left\{  1,2\right\}  }$ which coincide with $D_{f}$ on any commutative
subalgebra and is monotone non-increasing by application of CPTP maps. Then,
\ $D_{f}^{Q}\left(  \rho_{1}||\rho_{2}\right)  <\infty$ \ holds if
$\mathrm{supp}\,\rho_{1}$ $=\mathrm{supp}\,\,\rho_{2}$. Also, $D_{f}%
^{Q}\left(  \rho_{1}||\rho_{2}\right)  <\infty$ is equivalent to
(\ref{Wf-bounded}) if $\mathrm{supp}\,\rho_{1}$ $\not \subset \mathrm{supp}%
\,\,\rho_{2}$ and $\mathrm{supp}\,\rho_{2}$ $\not \subset \mathrm{supp}%
\,\,\rho_{1}$ hold, to (\ref{Wf-bounded-1}) if $\mathrm{supp}\rho_{1}%
\subset\mathrm{supp}\,\rho_{2}$ and $\ker\rho_{1}\neq\left\{  0\right\}  $
hold, and to (\ref{Wf-bounded-2})\ \ if $\mathrm{supp}\rho_{2}\subset
\mathrm{supp}\,\rho_{1}$ and $\ker\rho_{2}\neq\left\{  0\right\}  $ hold.
\end{theorem}

\section{Continuity of $D_{f}^{\min}$ and other versions of quantum
$f$-divergence}

\label{sec:continity}

Since $D_{f}^{\min}$ is jointly convex almost by definition, it is continuous
in the interior of $\mathrm{dom}\,D_{f}^{\min}$, or at the points where
$\rho_{1}>0$ and $\rho_{2}>0$. By applying well-known facts in convex
analysis, a weak version of "continuity at the boundary" is easily proved.

\begin{lemma}
$D_{f}^{\min}$ is a proper lower semi-continuous convex function which is
positively homogeneous.
\end{lemma}

\begin{proof}
Rewrite (\ref{D=sup-w-finite}) as
\[
D_{f}^{\min}\left(  \rho_{1}||\rho_{2}\right)  =\,\sup_{M}h_{M}\left(
\rho_{1},\rho_{2}\right)  ,
\]
where
\[
h_{M}\left(  \rho_{1},\rho_{2}\right)  :=\sup_{M}\sum_{\vec{w}\in
\mathrm{supp}\,M}\sum_{\theta\in\left\{  1,2\right\}  }w_{\theta}%
\,\mathrm{tr}\,\rho_{\theta}M_{\vec{w}},
\]
and $M$ moves for all over the POVM's whose support is a finite set in
$\mathcal{W}_{f}$. \ That $D_{f}^{\min}$ is proper, convex, and  positively
homogeneous is obvious from this expression. Also, recall that lower
semi-continuity is equivallent to closed epigraph. Thus, lower semi-continuity
of $D_{f}^{\min}$ follows from
\[
\bigcap_{M}\mathrm{epi}\,h_{M}=\mathrm{epi}\,D_{f}^{\min}%
\]
and closedness of $\mathrm{epi}\,h_{M}$. 
\end{proof}

\begin{theorem}
For any $\rho_{1}\geq0$ and $\rho_{2}\geq0$, and for any $X_{1},X_{2}\geq0$,
\begin{equation}
\lim_{s\downarrow0}D_{f}^{\min}\left(  \rho_{1}+sX_{1}||\rho_{2}%
+sX_{2}\right)  =D_{f}^{\min}\left(  \rho_{1}||\rho_{2}\right)
.\label{D-cont}%
\end{equation}

\end{theorem}

\begin{proof}
Since $D_{f}^{\min}$ is lower semicontinuous, it is continuous on any (finite
dimensional) simplex inside $\mathrm{dom}\,D_{f}^{\min}$ by Theorem\thinspace
10.2 of \cite{Rockafellar}. Applying this fact to the line segment connecting
$\left(  \rho_{1}+sX_{1},\rho_{2}+sX_{2}\right)  $ and $\left(  \rho_{1}%
,\rho_{2}\right)  $, we have the assertion.
\end{proof}

\bigskip

\begin{remark}
In may literature, various versions qunatum version of $f$-divergences
$D_{f}^{Q}$'s are defined for strictly positive operators, and then extended
to general positive operators (possibly with eigenvalue 0) in certain manners,
and analogues of (\ref{D-cont}) are proved exploiting various properties of
each $D_{f}^{Q}$.  

However, so far the author had observed, these extensions are all closure of
$D_{f}^{Q}$ : given  $D_{f}^{Q}$ whose effective domain $\mathrm{dom}%
\,D_{f}^{Q}$ is strictly positive operators, choose $\tilde{D}_{f}^{Q}$ by
$\mathrm{epi}\,\tilde{D}_{f}^{Q}=\mathrm{cl\,}\left(  \mathrm{epi}\,D_{f}%
^{Q}\right)  $. For such an extension,  analogue of (\ref{D-cont}) holds, by
Theorem\thinspace10.2 of \cite{Rockafellar}. 

\ By unique existence of the closure, analogue of (\ref{D-cont}) with some
fixed $X_{1}>0$, $X_{2}>0$ (they can be identity) can be used to define
$\tilde{D}_{f}^{Q}$.  Then (\ref{D-cont}) for all $X_{1}\geq0$, $X_{2}\geq0$,
convexity, and positive homogeneity holds.      
\end{remark}

\section{Expressions of \ \ $D_{f}^{\min}$ (II)}

\label{sec:sup-T}

In this section and in the most part of the remainder of the paper, we suppose
at least one of the followings is true.

\begin{description}
\item[(I)] $f^{\ast}$ is operator convex on $\mathrm{dom}\,f^{\ast}$

\item[(II)] $\hat{f}^{\ast}\left(  t\right)  $ ($=-f^{\ast-1}\left(
-t\right)  $) is operator convex on $\mathrm{dom}\,\hat{f}^{\ast}$
\end{description}

By (\ref{D=sup-t}) and Naimark's extension theorem,
\[
D_{f}^{\min}\left(  \rho_{1}||\rho_{2}\right)  =\sup_{E,V}\sum_{t\in
\mathrm{supp}\,E}\left\{  t\mathrm{tr}\,\rho_{1}V^{\dagger}E_{t}V-f^{\ast
}\left(  t\right)  \mathrm{tr}\,\rho_{2}V^{\dagger}E_{t}V\right\}  ,
\]
where $E$ is a projection valued measure (PVM, in short) on a Hilbert space
$\mathcal{K}\supset\mathcal{H}$, and $V$ is an isometry from $\mathcal{H}$
into $\mathcal{K}$. Then
\begin{align*}
D_{f}^{\min}\left(  \rho_{1}||\rho_{2}\right)   &  =\sup_{E,V}\left\{
\mathrm{tr}\,\rho_{1}V^{\dagger}\sum_{t\in\mathrm{supp}\,E}t\,E_{t}%
V-\mathrm{tr}\,\rho_{2}V^{\dagger}\sum_{t\in\mathrm{supp}\,E}f^{\ast}\left(
t\,\right)  E_{t}V\right\}  \\
&  =\sup_{T,V}\left\{  \mathrm{tr}\,\rho_{1}V^{\dagger}TV-\mathrm{tr}%
\,\rho_{2}V^{\dagger}f^{\ast}\left(  T\right)  V\right\}  ,
\end{align*}
where $T$ is a self-adjoint operator on the extended Hilbert space
$\mathcal{K}$ with $\mathrm{spec}\,T\subset$ $\mathrm{dom}\,f^{\ast}$. Here,
suppose (I) holds. Then by Jensen's inequality, we have
\begin{align*}
D_{f}^{\min}\left(  \rho_{1}||\rho_{2}\right)   &  \leq\sup_{T,V}\left\{
\mathrm{tr}\,\rho_{1}V^{\dagger}TV-\mathrm{tr}\,\rho_{2}f^{\ast}\left(
V^{\dagger}TV\right)  \right\}  \\
&  =\sup_{T,V}\left\{  \mathrm{tr}\,\rho_{1}V^{\dagger}TV-\mathrm{tr}%
\,\rho_{2}f^{\ast}\left(  V^{\dagger}TV\right)  \right\}  \\
&  =\sup_{T^{\prime}}\left\{  \mathrm{tr}\,\rho_{1}T^{\prime}-\mathrm{tr}%
\,\rho_{2}f^{\ast}\left(  T^{\prime}\right)  \right\}  ,
\end{align*}
where in the last end, $T^{\prime}$ is a self-adjoint operator on the original
Hilbert space $\mathcal{H}$ with $\mathrm{spec}\,T^{\prime}\subset$
$\mathrm{dom}\,f^{\ast}$. The opposite inequality is easily obtained by
restricting the measurements in (\ref{D=sup-t}) to PVM. Thus,
\begin{align}
D_{f}^{\min}\left(  \rho_{1}||\rho_{2}\right)   &  =\sup\left\{
\mathrm{tr}\,\rho_{1}T-\mathrm{tr}\,\rho_{2}f^{\ast}\left(  T\right)
;\mathrm{spec}\,T\subset\mathrm{dom}\,f^{\ast}\right\}  \label{Df-sup-T}\\
&  =\sup\left\{  \mathrm{tr}\,\rho_{1}f^{\ast-1}\left(  S\right)
-\mathrm{tr}\,\rho_{2}S;\,\mathrm{spec}\,S\subset\mathrm{dom}\,f^{\ast
-1}\right\}  .\label{Df-sup-S}%
\end{align}
The second identity is obtained by putting $S:=-f^{\ast}\left(  T\right)  $.
(Recall $f^{\ast}$ is monotone in its effective domain.) By assuming (II) and
using (\ref{D=sup-s}), we can also obtain the same identity. Summarizing the
above argument:

\begin{theorem}
\label{th:max-T}Let $f$ be a proper convex function with $\mathrm{dom}%
\,f\supset\left(  0,\infty\right)  $. If either (I) or (II) is true, we have
(\ref{Df-sup-T}) and (\ref{Df-sup-S}).
\end{theorem}

As we analyze later, $f_{\alpha}$ ($\alpha\leq\frac{1}{2}$) and
$f_{\mathrm{KL}_{2}}$ are the examples where (I) holds, and $f_{\alpha}$
($\alpha\geq\frac{1}{2}$) and $f_{\mathrm{KL}}$ are the examples where (II) holds.

\section{The stationary point}

\label{sec:sationary}

Throughout this section, unless otherwise mentioned, we suppose that the
condition (I) holds. This means that
\begin{equation}
G\left(  T\right)  :=\mathrm{tr}\,\rho_{1}T-\mathrm{tr}\,\rho_{2}f^{\ast
}\left(  T\right)  \label{def-G}%
\end{equation}
is concave in $T$, and its supremum equals $D_{f}^{\min}$. Moreover, we focus
on an easy case where the stationary point $T_{0}$ of $G$, or $T_{0}$ with
$\mathrm{d}G\left(  T_{0}\right)  /\mathrm{d}T=0$ exists in
\[
\mathrm{dom}\,G=\left\{  T\,;\mathrm{spec}\,T\,\subset\mathrm{dom}f^{\ast
}\right\}
\]
( $f^{\ast}$ is differentiable, being operator convex). Thus the supremum of
$G$ is achieved at $T_{0}$.

The Frechet derivative $\mathrm{D}f^{\ast}\left(  T\right)  $ of $f^{\ast}$
i.e., a linear transform in $\mathcal{B}\left(  \mathcal{H}\right)  $ with%
\[
\left\Vert f^{\ast}\left(  T+X\right)  -f^{\ast}\left(  T\right)
-\mathrm{D}f^{\ast}\left(  T\right)  \left(  X\right)  \right\Vert
_{2}=o\left(  \left\Vert X\right\Vert _{2}\right)
\]
is given by, in the basis which diagonalizes $T$,
\begin{equation}
\,\mathrm{D}f^{\ast}\left(  T\right)  \left(  X\right)  =\left[
f^{\ast\left[  1\right]  }\left(  t_{i},t_{j}\right)  X_{i,j}\right]
,\label{Frechet}%
\end{equation}
where $t_{i}$ ($i=1,\cdots$) are eigenvalues of $T$, and
\[
f^{\left[  1\right]  }\left(  t,t^{\prime}\right)  :=\left\{
\begin{array}
[c]{cc}%
\frac{f\left(  t\right)  -f\left(  t^{\prime}\right)  }{t-t^{\prime}}, &
\left(  t\neq t^{\prime}\right)  ,\\
f^{\prime}\left(  t\right)  , & \left(  t=t^{\prime}\right)  .
\end{array}
\right.
\]
A consequence of this formula is $\,\mathrm{D}f^{\ast}\left(  T\right)
\left(  \cdot\right)  $ is self-adjoint with respect to the inner product
$\mathrm{tr}\,XY$,
\begin{align*}
\mathrm{tr}\,Y\mathrm{D}f^{\ast}\left(  T\right)  \left(  X\right)   &
=\sum_{i,j}\overline{\rho_{2,i,j}}f^{\ast\left[  1\right]  }\left(
t_{i},t_{j}\right)  X_{i,j}\\
&  =\sum_{i,j}\overline{f^{\ast\left[  1\right]  }\left(  t_{i},t_{j}\right)
\rho_{2,i,j}}X_{i,j}\\
&  =\mathrm{tr}\,X\mathrm{D}f^{\ast}\left(  T\right)  \left(  Y\right)  .
\end{align*}

With these definitions and assumptions, we now proceed to the analysis of the
maximal point of $G\left(  T\right)  $. Most tractable case is that there is a
stationary point of $G\left(  T\right)  $ in $\mathrm{dom}\,G$. If
\begin{align*}
\left.  \frac{\mathrm{d}\,G\left(  T_{0}+sX\right)  }{\mathrm{d}%
\,s}\right\vert _{s=0} &  =\mathrm{tr}\,X\rho_{1}-\mathrm{tr}\,\left\{
\rho_{2}\,\mathrm{D}f^{\ast}\left(  T_{0}\right)  \left(  X\right)  \right\}
\\
&  =\mathrm{tr}\,X\left(  \rho_{1}-\mathrm{D}f^{\ast}\left(  T_{0}\right)
\left(  \rho_{2}\right)  \right)  =0
\end{align*}
holds for any Hermitian matrix $X$, $T_{0}$ achieves maximum. (Here we used
the fact that $\mathrm{D}f^{\ast}\left(  T_{0}\right)  \left(  \cdot\right)  $
is self-conjugate.) Thus, we have
\begin{equation}
\rho_{1}=\mathrm{D}f^{\ast}\left(  T_{0}\right)  \left(  \rho_{2}\right)
\label{rho=Df}%
\end{equation}

Therefore,
\begin{align}
D_{f}^{\min}\left(  \rho_{1}||\rho_{2}\right)   &  =\mathrm{tr}\,\rho_{1}%
T_{0}-\mathrm{tr}\,\rho_{2}f^{\ast}\left(  T_{0}\right)  \nonumber\\
&  =\mathrm{tr}\,T_{0}\,\mathrm{D}f^{\ast}\left(  T_{0}\right)  \left(
\rho_{2}\right)  -\mathrm{tr}\,\rho_{2}f^{\ast}\left(  T_{0}\right)
\nonumber\\
&  =\mathrm{tr}\,\mathrm{D}f^{\ast}\left(  T_{0}\right)  \left(  T_{0}\right)
\rho_{2}-\mathrm{tr}\,\rho_{2}f^{\ast}\left(  T_{0}\right)  \nonumber\\
&  =\mathrm{tr}\,\left\{  \mathrm{D}f^{\ast}\left(  T_{0}\right)  \left(
T_{0}\right)  -f^{\ast}\left(  T_{0}\right)  \right\}  \rho_{2}\nonumber\\
&  =\mathrm{tr}\,\left\{  T_{0}\cdot f^{\ast^{\prime}}\left(  T_{0}\right)
-f^{\ast}\left(  T_{0}\right)  \right\}  \rho_{2}\nonumber\\
&  =\mathrm{tr}\,\,f\left(  f^{\ast^{\prime}}\left(  T_{0}\right)  \right)
\rho_{2},\label{Dfmin=tr-f-rho}%
\end{align}
Here, in the third identity holds since $\mathrm{D}f^{\ast}\left(
T_{0}\right)  \left(  \cdot\right)  $ is self-adjoint, the fifth identity
holds due to (\ref{Frechet}), and the last identity is due to $f=f^{\ast\ast}$.

\ Combinig above argument with (\ref{suppM}), we have:

\begin{theorem}
\label{th:T-domain}Suppose $f$ is a proper lower semicontinuous convex
function with $\mathrm{dom}\,f\supset\left(  0,\infty\right)  $. Suppose also
the assumption (I) is true. Then, the $T$ in (\ref{Df-sup-T}) can be
restricted to the set of all the Hermitian operators with $\mathrm{spec}%
\,T\subset\left[  t_{\ast}^{\prime},t_{\ast}\right]  $, $t_{\ast}^{\prime}$
and $t_{\ast}$ are as of Definition\thinspace\ref{def:b-t}. Also, if $\left[
t_{\ast}^{\prime},t_{\ast}\right]  $ is contained in the interior of
$\mathrm{dom}\,f^{\ast}$, the solution $T_{0}$ to (\ref{rho=Df}) exists and
achieves the supremum.
\end{theorem}

So far, we had supposed the assumption (I) is true. Now let us consider the
case where (II) holds and (I) does not. A trivial approach is to exchange
$\rho_{0}$ and $\rho_{1}$, and apply all the analysis replacing $f$ by
$\hat{f}$. This means the change of variable from $T$ to $S:=-f^{\ast}\left(
T\right)  $. But sometimes, the use of the variable $T$ is more preferable for
technical reasons. In such cases, still we can use (\ref{rho=Df}), because of
the following reason. Since the assumption (II) says that $\hat{f}^{\ast
}\left(  t\right)  =-f^{\ast-1}\left(  t\right)  $ is operator convex, it is
continuously differentiable. Thus the stable point with respect to $S$ is also
a stable point with respect to $T$.

\begin{remark}
The linear map $\mathrm{D}f^{\ast}\left(  T_{0}\right)  \left(  \cdot\right)
$ in fact is completely positive, if $f$ satisfies (I), $\mathrm{dom}$
$f^{\ast}$ is not bounded from below, and $f^{\ast}\left(  -\infty\right)
>-\infty$. \ Since $f^{\ast}$ is operator monotone by Lemma\thinspace
\ref{lem:monotone} below, the matrix $\left[  f^{\ast\left[  1\right]
}\left(  t_{i},t_{j}\right)  \right]  _{i,j=1}^{n}$ is positive for any $n$
and $t_{1}$,$t_{2}$,$\cdots$,$t_{n}\in\mathrm{dom}$ $f^{\ast}$, by
Theorem\thinspace2.4.3 of \cite{Hiai}. Therefore, there is a complex numbers
$\beta_{i,1}$, $\beta_{i,2}$, $\cdots$ with $f^{\ast\left[  1\right]  }\left(
t_{i},t_{j}\right)  =\sum_{\kappa}\beta_{i,\kappa}\overline{\beta_{j,\kappa}}%
$, and thus
\begin{align*}
\mathrm{D}f^{\ast}\left(  T_{0}\right)  \left(  X\right)   &  =\sum
_{i,j}f^{\ast\left[  1\right]  }\left(  t_{i},t_{j}\right)  \left\vert
e_{i}\right\rangle \left\langle e_{i}\right\vert X\left\vert e_{j}%
\right\rangle \left\langle e_{j}\right\vert \\
&  =\sum_{\kappa}\left(  \sum_{i}\beta_{i,\kappa}\left\vert e_{i}\right\rangle
\left\langle e_{i}\right\vert \right)  X\left(  \sum_{j}\overline
{\beta_{j,\kappa}}\left\vert e_{j}\right\rangle \left\langle e_{j}\right\vert
\right)  ,
\end{align*}
where $T_{0}=\sum_{i}t_{i}\left\vert e_{i}\right\rangle \left\langle
e_{i}\right\vert $. \ 
\end{remark}

\begin{remark}
By (\ref{rho=Df}) and (\ref{Dfmin=tr-f-rho}), $\mathrm{D}f^{\ast}\left(
T_{0}\right)  \left(  \cdot\right)  $ and/or $f^{\ast^{\prime}}\left(
T_{0}\right)  $ may be viewed as a non-commutative version of Radon-Nikodym
derivative $\mathrm{d}P_{1}/\mathrm{d}P_{2}$.
\end{remark}

\section{Non-full-rank states}

\label{sec:non-full-rank}

In this section we study thes case where $\ker\,\rho_{1}$ is non-trivial, with
the additional assumptions that the condition (I) holds,  and $\mathrm{dom}$
$f^{\ast}$ is not bounded from below. By the following lemma, whose proof is
given in Appendix\thinspace\ref{sec:proof-lem-monotone},  $f^{\ast}$ in fact
is operator monotone increasing.

\begin{lemma}
\label{lem:monotone}Suppose $f$ is proper, lower semicontinuous, convex,
canonical and $\mathrm{dom}\,f\supset(0,\infty)$. Suppose also that  $f^{\ast
}$ is operator convex, that  $\mathrm{dom}$ $f^{\ast}$ is not bounded from
below, and that $f^{\ast}\left(  -\infty\right)  >-\infty$. Then $f^{\ast}$ is
operator monotone.
\end{lemma}

\begin{remark}
$\mathrm{dom}$ $f^{\ast}$ is not bounded from below if and only if right
derivative of $f$ at $\lambda=0$ is $-\infty$. In this case $f^{\ast}\left(
-\infty\right)  =f\left(  0\right)  $.
\end{remark}

Below, $\pi_{\mathcal{H}^{\prime}}$ denotes the projection onto $\mathcal{H}%
^{\prime}$and  $\pi_{\rho}$ ($\rho\geq0$) is short for $\pi_{\mathrm{supp}%
\,\rho}$. Also,  $T_{\mathcal{H}^{\prime}}$ is the restriction of
$\pi_{\mathcal{H}^{\prime}}T\pi_{\mathcal{H}^{\prime}}$ on $\mathcal{H}%
^{\prime}$ .

\begin{lemma}
\label{lem:f*(T+sX)->}Suppose the domain of a function $h$ is unbounded from
below and $h\left(  -\infty\right)  $ is finite. Then, if $X\leq0$,
\begin{equation}
\lim_{s\rightarrow\infty}h\left(  T+sX\right)  =\,h\left(  T_{\ker\,X}\right)
\pi_{\ker X}+h\left(  -\infty\right)  \pi_{\mathrm{supp}X}. \label{f*(T+sX)->}%
\end{equation}

\end{lemma}

\begin{proof}
Let $r_{s}^{\kappa}$ and $\left\vert \varphi_{s}^{\kappa}\right\rangle $
($\left\Vert \varphi_{s}^{\kappa}\right\Vert =1$) be the $\kappa$-th
eigenvalue and eigenvector of $T+sX$, respectively. Since
\[
\lim_{s\rightarrow\infty}\left(  s^{-1}T+X\right)  =X,
\]
$\lim_{s\rightarrow\infty}r_{s}^{\kappa}/s:=\xi^{\kappa}$ and $\lim
_{s\rightarrow\infty}\left\vert \varphi_{s}^{\kappa}\right\rangle =:\left\vert
\varphi_{\infty}^{\kappa}\right\rangle $ gives a complete set of eigenvalues
and eigenvectors of $X$. Without loss of generality, let $\xi^{\kappa}$ be 0
if $\kappa\leq\dim\ker X$, and negative otherwise.

Let $t$ be an eigenvalue of $T_{\ker\,X}$ \ and denote by $\mathcal{H}_{t}$
and $\pi_{t}$  the corresponding eigenspace and the projection onto it.
Observe, if $\kappa\leq\dim\ker\,X$, \ we have
\begin{align*}
&  r_{s}^{\kappa}\pi_{t}\left\vert \varphi_{s}^{\kappa}\right\rangle \\
&  =\pi_{t}\left(  T+sX\right)  \left\vert \varphi_{s}^{\kappa}\right\rangle
=\pi_{t}T\left\vert \varphi_{s}^{\kappa}\right\rangle \\
&  =\pi_{t}\left(  \pi_{\ker X}+\pi_{\mathrm{supp}X}\right)  T\left\vert
\varphi_{s}^{\kappa}\right\rangle =\pi_{t}\pi_{\ker X}T\left\vert \varphi
_{s}^{\kappa}\right\rangle \\
&  =\left(  \pi_{t}\pi_{\ker X}T\pi_{\ker X}+\pi_{t}\pi_{\ker X}%
T\pi_{\mathrm{supp}X}\right)  \left\vert \varphi_{s}^{\kappa}\right\rangle \\
&  =t\pi_{\mathcal{H}_{t}}\left\vert \varphi_{s}^{\kappa}\right\rangle
+\pi_{\mathcal{H}_{t}}T\pi_{\mathrm{supp}X}\left\vert \varphi_{s}^{\kappa
}\right\rangle ,
\end{align*}
where the second and the foruth identity is by $\pi_{t}X=0$ (recall that
$\mathcal{H}_{t}$ is a subspace of $\ker X$), and the last identity is by
$\pi_{t}T_{\ker\,X}=t\pi_{t}$. Thus we have
\begin{align*}
\lim_{s\rightarrow\infty}\left(  r_{s}^{\kappa}-t\right)  \pi_{\mathcal{H}%
_{t}}\left\vert \varphi_{s}^{\kappa}\right\rangle  &  =\lim_{s\rightarrow
\infty}\pi_{\mathcal{H}_{t}}T\pi_{\mathrm{supp}X}\left\vert \varphi
_{s}^{\kappa}\right\rangle \\
&  =0.
\end{align*}

Therefore, if $\pi_{\mathcal{H}_{t}}\left\vert \varphi_{\infty}^{\kappa
}\right\rangle $ $\neq0$ holds, $r_{s}^{\kappa}\rightarrow t$ holds, implying
that $\left\vert \varphi_{\infty}^{\kappa}\right\rangle $ is a member of
$\mathcal{H}_{t}$. Since $\left\vert \varphi_{\infty}^{\kappa}\right\rangle $
($\kappa\leq\dim\ker X$) is a member of $\ker\,X$, which is the direct sum of
all $\mathcal{H}_{t}$'s, it overlaps with, and thus a member of, at least one
of $\mathcal{H}_{t}$'s. Therefore, recalling that $\left\{  \left\vert
\varphi_{\infty}^{\kappa}\right\rangle ;\kappa\leq\dim\ker X\right\}  $ is a
CONS of $\ker\,X$,
\[
\sum_{\kappa:\lim_{s\rightarrow\infty}r_{s}^{\kappa}=t}\left\vert
\varphi_{\infty}^{\kappa}\right\rangle \left\langle \varphi_{\infty}^{\kappa
}\right\vert =\pi_{\mathcal{H}_{t}}.
\]
Therefore,
\begin{align*}
&  \lim_{s\rightarrow\infty}h\left(  T+sX\right)  \\
= &  \lim_{s\rightarrow\infty}\sum_{\kappa\leq\dim\ker X}h\left(
r_{s}^{\kappa}\right)  \left\vert \varphi_{s}^{\kappa}\right\rangle
\left\langle \varphi_{s}^{\kappa}\right\vert +\lim_{s\rightarrow\infty}%
\sum_{\kappa>\dim\ker X}h\left(  r_{s}^{\kappa}\right)  \left\vert \varphi
_{s}^{\kappa}\right\rangle \left\langle \varphi_{s}^{\kappa}\right\vert \\
&  =\sum_{t}h\left(  t\right)  \pi_{\mathcal{H}_{t}}+\sum_{\kappa>\dim\ker
X}\lim_{s\rightarrow\infty}h\left(  s\xi^{\kappa}\right)  \left\vert
\varphi_{\infty}^{\kappa}\right\rangle \left\langle \varphi_{\infty}^{\kappa
}\right\vert \\
&  =h\left(  T_{\ker\,X}\right)  \pi_{\ker X}+h\left(  -\infty\right)
\pi_{\mathrm{supp}X}.
\end{align*}

\end{proof}

\begin{theorem}
\label{th:with-kernel}Suppose $f$ is proper, lower semicontinuous, convex,
canonical and $\mathrm{dom}\,f\supset(0,\infty)$. Suppose also $f^{\ast}$ is
operator convex and $\mathrm{dom}$ $f^{\ast}$ is not bounded from below. Then
if $\rho_{1}$ is not full-rank,
\begin{equation}
D_{f}^{\min}\left(  \rho_{1}||\rho_{2}\right)  =D_{f}^{\min}\left(  \rho
_{1}||\pi_{\rho_{1}}\rho_{2}\pi_{\rho_{1}}\right)  +f\left(  0\right)
\mathrm{tr}\,\rho_{2}\left(  \mathbf{1}-\pi_{\rho_{1}}\right)
.\label{with-kernel}%
\end{equation}
In addition, the measurement achives the maximum is composition of the
projective measurement $\left\{  \pi_{\rho_{1}},\mathbf{1}-\pi_{\rho_{1}%
}\right\}  $ followed by a measurement on $\mathrm{supp}\rho$.
\end{theorem}

By this theorem, the problem reduces to maximization of $G\left(  T\right)  $
as of (\ref{def-G})\ for all $T$ to on the support of $\rho_{1}$. Especially,
if $\rho_{1}$ is rank-1 state, $\rho=\left\vert \varphi\right\rangle
\left\langle \varphi\right\vert $, and $\left\langle \varphi_{1}\right\vert
\rho_{2}\left\vert \varphi_{1}\right\rangle \neq0$,
\begin{align}
D_{f}^{\min}\left(  \left\vert \varphi_{1}\right\rangle \left\langle
\varphi_{1}\right\vert \,||\rho_{2}\right)   &  =\sup_{t\in\mathrm{dom}%
\,f^{\ast}}\left\{  t-\left\langle \varphi_{1}\right\vert \rho_{2}\left\vert
\varphi_{1}\right\rangle \,f^{\ast}\left(  t\right)  \right\}  +f\left(
0\right)  \left(  1-\left\langle \varphi_{1}\right\vert \rho_{2}\left\vert
\varphi_{1}\right\rangle \right)  \nonumber\\
&  =\hat{f}\left(  \left\langle \varphi_{1}\right\vert \rho_{2}\left\vert
\varphi_{1}\right\rangle \right)  +f\left(  0\right)  \left(  1-\left\langle
\varphi_{1}\right\vert \rho_{2}\left\vert \varphi_{1}\right\rangle \right)
.\label{Df-pure}%
\end{align}
In addition, the measurement achives the maximum is composition of the
projective measurement $\left\{  \left\vert \varphi_{1}\right\rangle
\left\langle \varphi_{1}\right\vert ,\mathbf{1}-\left\vert \varphi
_{1}\right\rangle \left\langle \varphi_{1}\right\vert \right\}  $.

\begin{proof}
By Theorem\thinspace\ref{th:Df-finite}, we only have to prove the assertion
for the case $-\infty<f^{\ast}\left(  -\infty\right)  $ (otherwise,
$D_{f}^{\min}\left(  \rho_{1}||\rho_{2}\right)  =\infty$). By Lemma\thinspace
\ref{lem:monotone}, $f^{\ast}$ is operator monotone. Thus, if $X\leq0$ is
supported on $\ker\,\rho_{1}$, $G\left(  T+sX\right)  $, where $G$ is as of
(\ref{def-G}), is non decreasing in $s$. Therefore,
\[
\sup_{s}G\left(  T+sX\right)  =\lim_{s\rightarrow\infty}G\left(  T+sX\right)
.
\]
Hence, by Lemma\thinspace\ref{lem:f*(T+sX)->}, we have
\begin{align*}
\sup_{X\leq0:\,\mathrm{supp\,}X=\text{ }\ker\,\rho_{1}}G\left(  T+X\right)
&  =\mathrm{tr}\,\rho_{1}T_{\mathrm{supp}\rho_{1}}-\mathrm{tr}\,\left(
\rho_{2}\right)  _{\mathrm{supp}\rho_{1}}f^{\ast}\left(  T_{\mathrm{supp}%
\rho_{1}}\right)  \\
&  -f^{\ast}\left(  -\infty\right)  \mathrm{tr}\,\rho_{2}\left(
\mathbf{1}-\pi_{\rho_{1}}\right)  .
\end{align*}

Since $f$ is canonical, $f$ is affine on the negative half-line or infinite.
If the former is the case, $\mathrm{dom}\,f^{\ast}$ is bounded from below.
Thus, $\mathrm{dom}\,f$ does not extend to negative half-line. Hence,
$f^{\ast}\left(  -\infty\right)  =-f\left(  0\right)  $. Therefore, we have
(\ref{with-kernel}).
\end{proof}

\section{Examples}

\label{sec:examples}

\subsection{Renyi-type, $f_{\alpha}\left(  \lambda\right)  $}

$f_{\alpha}\left(  \lambda\right)  $  is defined by (\ref{Renyi}) on the
non-negative half line. On the negative half-line, it is canonically extended.
We omit the cases of $\alpha=0$ and $\alpha=1$. The relation
\begin{equation}
\hat{f}_{\alpha}\left(  \lambda\right)  =f_{1-\alpha}\left(  \lambda\right)
,\label{hatf=f}%
\end{equation}
turns out to be quite useful.

If $0<\alpha<1$,
\[
f_{\alpha}^{\ast}\left(  t\right)  =\left\{
\begin{array}
[c]{cc}%
\infty, & \left(  t>0\right)  ,\\
\left(  1-\alpha\right)  \alpha^{\frac{\alpha}{1-\alpha}}\left(  -t\right)
^{\frac{-\alpha}{1-\alpha}}, & \left(  t\leq0\right)  .
\end{array}
\right.
\]
Thus, if $0<\alpha\leq\frac{1}{2}$, the condition (I) satisfied, and if
$\frac{1}{2}\leq\alpha<1$, the condition (II) is satisfied. If $\alpha<0$,%
\[
f_{\alpha}^{\ast}\left(  t\right)  =\left\{
\begin{array}
[c]{cc}%
\infty, & \left(  t\geq0\right)  ,\\
\left(  \alpha-1\right)  \left(  -\alpha\right)  ^{\frac{\alpha}{1-\alpha}%
}\left(  -t\right)  ^{\frac{-\alpha}{1-\alpha}}, & \left(  t<0\right)  ,
\end{array}
\right.
\]
and the condition (I) is satisfied. Thus, \ by (\ref{hatf=f}), the condition
(II) is satisfied for $\alpha>1$. In this case,
\[
f_{\alpha}^{\ast}\left(  t\right)  =\left\{
\begin{array}
[c]{cc}%
\left(  \alpha-1\right)  \alpha^{\frac{-\alpha}{\alpha-1}}t^{\frac{\alpha
}{\alpha-1}}, & \left(  t>0\right)  ,\\
\infty, & \left(  t\leq0\right)  ,
\end{array}
\right.
\]
implying that the condition (I) \ is also satisfied for $\alpha\geq2$.

For all the values of $\alpha$ ($\neq0,1$), $f_{\alpha}^{\ast\prime}\left(
t\right)  $ moves all over the positive half line $\left(  0,\infty\right)  $.
Thus, by Theorem\thinspace\ref{th:T-domain}, the supremum is achieved by
$T_{0}$ with (\ref{rho=Df}).

In the case of $\alpha=-1$ and $2$, we can solve the problem "explicitly".
Observe $f_{2}^{\ast}$ is operator convex on $\mathrm{dom}\,f_{2}^{\ast}$.%

\[
f_{2}^{\ast}\left(  t\right)  =\left\{
\begin{array}
[c]{cc}%
\frac{1}{4}t^{2}, & \left(  t\geq0\right)  ,\\
\infty, & \left(  t<0\right)  .
\end{array}
\right.  ,
\]
By (\ref{rho=Df}), $T_{0}$ satisfies Lyapunov equation
\begin{equation}
\mathrm{D}f_{2}^{\ast}\left(  T_{0}\right)  \left(  \rho_{2}\right)  =\frac
{1}{4}\left(  T_{0}\rho_{2}+\rho_{2}T_{0}\right)  =\rho_{1}.\label{t0-f2}%
\end{equation}
If $\mathrm{supp}\rho_{1}\subset\mathrm{supp}\,\rho_{2}$, this equation about
$T_{0}$ has a solution, namely
\[
T_{0}=4\int_{-\infty}^{0}e^{s\rho_{2}}\rho_{1}e^{s\rho_{2}}\mathrm{d}s\geq0,
\]
and in the basis where $\rho_{2}$ is diagonal,
\[
T_{0,i,j}=\frac{4}{\rho_{2,i,i}+\rho_{2,j,j}}\rho_{1,i,j}.
\]
Thus, this solution has spectrum in $\mathrm{dom}\,f_{2}^{\ast}$. By
(\ref{Dfmin=tr-f-rho}),%
\begin{align*}
D_{f_{-1}}^{\min}\left(  \rho_{2}||\rho_{1}\right)   &  =D_{f_{2}}^{\min
}\left(  \rho_{1}||\rho_{2}\right)  =\mathrm{tr}\,\rho_{2}\left(  \frac{1}%
{2}T_{0}\right)  ^{2}=\frac{1}{4}\mathrm{tr}\,\rho_{2}T_{0}^{2}\\
&  =\frac{1}{2}\mathrm{tr}\,\rho_{1}T_{0}\\
&  =2\mathrm{tr}\,\rho_{1}\int_{-\infty}^{0}e^{s\rho_{2}}\rho_{1}e^{s\rho_{2}%
}\mathrm{d}s\\
&  =2\sum_{i,j}\frac{1}{\rho_{2,i,i}+\rho_{2,j,j}}\left\vert \rho
_{1,i,j}\right\vert ^{2}.
\end{align*}

If $\alpha=\frac{1}{2}$,
\[
f_{\frac{1}{2}}^{\ast}\left(  t\right)  =\left\{
\begin{array}
[c]{cc}%
\infty, & \left(  t>0\right)  ,\\
-\frac{1}{4}t^{-1}, & \left(  t\leq0\right)  ,
\end{array}
\right.
\]
and
\begin{equation}
\mathrm{D}f_{\frac{1}{2}}^{\ast}\left(  T_{0}\right)  \left(  \rho_{2}\right)
=\frac{1}{4}T_{0}^{-1}\rho_{2}T_{0}^{-1}=\rho_{1}.\label{t0-f1/2}%
\end{equation}
Thus%
\[
T_{0}^{-1}=-2\rho_{2}^{-1/2}\sqrt{\rho_{2}^{1/2}\rho_{1}\rho_{2}^{1/2}}%
\rho_{2}^{-1/2},
\]
and
\begin{align*}
D_{f_{1/2}}^{\min}\left(  \rho_{1}||\rho_{2}\right)   &  =-\mathrm{tr}%
\,\rho_{2}\left(  \frac{1}{4}T_{0}^{-2}\right)  ^{1/2}=-\frac{1}{2}%
\mathrm{tr}\,\rho_{2}\left(  -T_{0}^{-1}\right)  \\
&  =-\mathrm{tr}\,\sqrt{\rho_{2}^{1/2}\rho_{1}\rho_{2}^{1/2}},
\end{align*}
which is $-1$ times the fidelity of $\rho_{1}$ and $\rho_{2}$, as expected.

\begin{remark}
$\frac{1}{4}T_{0}$, where $T_{0}$ is as of (\ref{t0-f2}), is called `linear
Radon-Nikodym derivative' \cite{Sakai:1971}, and $-\frac{1}{2}T_{0}^{-1}$,
where $T_{0}$ is as of (\ref{t0-f1/2}), is `called quadratic Radon-Nikodym
derivative' \cite{Sakai:1965}.
\end{remark}

By Theorem \ref{th:Df-finite}, $D_{f_{\alpha}}^{\min}\left(  \rho_{1}%
||\rho_{2}\right)  <\infty$ for any $\rho_{1}$ and $\rho_{2}$ with
$\mathrm{supp}\,\rho_{1}\supset\mathrm{supp}\,\rho_{2}$ ($\mathrm{supp}%
\,\rho_{1}\subset\mathrm{supp}\,\rho_{2}$, resp.) if $\alpha<0$ (if $\alpha
>1$, resp.). If $0<\alpha<1$, $D_{f_{\alpha}}^{\min}\left(  \rho_{1}||\rho
_{2}\right)  $ is finite for all $\rho_{1}\geq0$ and $\rho_{2}\geq0$.

If $0<\alpha\leq\frac{1}{2}$, $f_{\alpha}^{\ast}$ is operator monotone, and
thus by (\ref{Df-pure}),
\begin{align*}
D_{f_{\alpha}}^{\min}\left(  \left\vert \varphi_{1}\right\rangle \left\langle
\varphi_{1}\right\vert ||\rho_{2}\right)   &  =\hat{f}_{\alpha}\left(
\left\langle \varphi_{1}\right\vert \rho_{2}\left\vert \varphi_{1}%
\right\rangle \right)  -f_{\alpha}^{\ast}\left(  -\infty\right)  \left(
1-\left\langle \varphi_{1}\right\vert \rho_{2}\left\vert \varphi
_{1}\right\rangle \right)  \\
&  =\,-\left\langle \varphi_{1}\right\vert \rho_{2}\left\vert \varphi
_{1}\right\rangle ^{1-\alpha},
\end{align*}
and
\begin{equation}
D_{f_{\alpha}}^{\min}\left(  \left\vert \varphi_{1}\right\rangle \left\langle
\varphi_{1}\right\vert \,||\,\left\vert \varphi_{2}\right\rangle \left\langle
\varphi_{2}\right\vert \right)  =-\left\vert \left\langle \varphi_{1}\right.
\left\vert \varphi_{2}\right\rangle \right\vert ^{2\left(  1-\alpha\right)
}.\label{D-reny-pure-1}%
\end{equation}
This means, if $\frac{1}{2}\leq\alpha<1$, using $\hat{f}_{\alpha}=f_{1-\alpha
}$,%
\begin{equation}
D_{f_{\alpha}}^{\min}\left(  \left\vert \varphi_{1}\right\rangle \left\langle
\varphi_{1}\right\vert \,||\,\left\vert \varphi_{2}\right\rangle \left\langle
\varphi_{2}\right\vert \right)  =-\left\vert \left\langle \varphi_{1}\right.
\left\vert \varphi_{2}\right\rangle \right\vert ^{2\alpha}%
\label{D-reny-pure-2}%
\end{equation}

\subsection{On Chernoff and Hoeffding bound}

Above mentioned explicit expression ( (\ref{D-reny-pure-1}) and
(\ref{D-reny-pure-2}) ) of Renyi-type quantity of pure states gives another
way to compute celebrated quantum Chernoff bound and Hoeffding bound, whose
classical counter part is
\begin{align*}
C\left(  p_{1}||p_{2}\right)   &  :=\sup\varliminf_{n\rightarrow\infty}%
\frac{-1}{n}\ln\left(  \eta_{1,n}+\eta_{2,n}\right)  \\
&  =\sup_{0<\alpha<1}\left\{  -\ln\left(  -D_{f_{\alpha}}\left(  p_{1}%
||p_{2}\right)  \right)  \right\}
\end{align*}
and
\begin{align*}
H_{r}\left(  p_{1}||p_{2}\right)   &  :=\sup\left\{  -\lim_{n\rightarrow
\infty}\frac{1}{n}\ln\eta_{1,n};\,\varlimsup_{n\rightarrow\infty}\frac{1}%
{n}\eta_{2,n}\leq-r\right\}  \\
= &  \sup_{0<\alpha<1}-\frac{\alpha r}{1-\alpha}-\frac{1}{\left(
1-\alpha\right)  }\ln\left(  -D_{f_{\alpha}}\left(  p_{1}||p_{2}\right)
\right)  ,
\end{align*}
respectively. Here, $\eta_{1,n}$ ($\eta_{2,n}$, resp.) is the probability that
the test mistakenly judges the true distribution as being $p_{2}^{\otimes n}$
($p_{1}^{\otimes n}$, respp.), while it is in fact $p_{1}^{\otimes n}$%
($p_{2}^{\otimes n}$, resp.). Their quantum counterparts are, in case that
both states are pure,
\begin{align}
C\left(  \rho_{1}||\rho_{2}\right)   &  =\sup_{0<\alpha<1}\left\{
-2\ln\left\vert \left\langle \varphi_{1}\right.  \left\vert \varphi
_{2}\right\rangle \right\vert \right\}  \label{C-pure-1}\\
&  =-2\ln\left\vert \left\langle \varphi_{1}\right.  \left\vert \varphi
_{2}\right\rangle \right\vert ,\nonumber\\
H_{r}\left(  p_{1}||p_{2}\right)   &  =\sup_{0<\alpha<1}-\frac{\alpha
r}{1-\alpha}-\frac{2}{\left(  1-\alpha\right)  }\ln\left(  \left\vert
\left\langle \varphi_{1}\right.  \left\vert \varphi_{2}\right\rangle
\right\vert \right)  \label{H-pure-1}\\
&  =\left\{
\begin{array}
[c]{cc}%
-2\ln\left\vert \left\langle \varphi_{1}\right.  \left\vert \varphi
_{2}\right\rangle \right\vert , & r\leq-2\ln\left\vert \left\langle
\varphi_{1}\right.  \left\vert \varphi_{2}\right\rangle \right\vert ,\\
\infty, & r>-2\ln\left\vert \left\langle \varphi_{1}\right.  \left\vert
\varphi_{2}\right\rangle \right\vert .
\end{array}
\right.  \nonumber
\end{align}
(See \cite{Audeneart}\cite{HiaiMosonyiPetzBeny}\cite{Nagaoka}.)

We confirm the achievability part of these celebrated results in the case that
each $\rho_{\theta}$ ($\theta=1,2$) is a pure state, $\left\vert
\varphi_{\theta}\right\rangle \left\langle \varphi_{\theta}\right\vert $,
using (\ref{D-reny-pure-1}) and (\ref{D-reny-pure-2}). we have
\begin{align}
C\left(  \rho_{1}||\rho_{2}\right)    & \geq\lim_{n\rightarrow\infty}\sup
_{M}\sup_{0<\alpha<1}\left\{  -\frac{1}{n}\ln\left(  -D_{f_{\alpha}}\left(
P_{\rho_{1}^{\otimes n}}^{M}||P_{\rho_{2}^{\otimes n}}^{M}\right)  \right)
\right\}  \nonumber\\
& =\lim_{n\rightarrow\infty}\sup_{0<\alpha<1}\sup_{M}\left\{  -\frac{1}{n}%
\ln\left(  -D_{f_{\alpha}}\left(  P_{\rho_{1}^{\otimes n}}^{M}||P_{\rho
_{2}^{\otimes n}}^{M}\right)  \right)  \right\}  \nonumber\\
& =\max\left\{  \sup_{0<\alpha\leq\frac{1}{2}}-2\left(  1-\alpha\right)
\ln\left\vert \left\langle \varphi_{1}\right.  \left\vert \varphi
_{2}\right\rangle \right\vert ,\sup_{\frac{1}{2}<\alpha<1}-2\alpha
\ln\left\vert \left\langle \varphi_{1}\right.  \left\vert \varphi
_{2}\right\rangle \right\vert \right\}  \label{C-pure-2}\\
& =-2\ln\left\vert \left\langle \varphi_{1}\right.  \left\vert \varphi
_{2}\right\rangle \right\vert ,\nonumber
\end{align}
and
\begin{align}
&  H_{r}\left(  p_{1}||p_{2}\right)  \nonumber\\
&  \geq\max\left\{  \sup_{0<\alpha\leq\frac{1}{2}}-\frac{\alpha r}{1-\alpha
}-2\ln\left\vert \left\langle \varphi_{1}\right.  \left\vert \varphi
_{2}\right\rangle \right\vert ,\sup_{\frac{1}{2}<\alpha<1}-\frac{\alpha
r}{1-\alpha}-\frac{2\alpha}{1-\alpha}\ln\left\vert \left\langle \varphi
_{1}\right.  \left\vert \varphi_{2}\right\rangle \right\vert \right\}
\label{H-pure-2}\\
&  =\left\{
\begin{array}
[c]{cc}%
-2\ln\left\vert \left\langle \varphi_{1}\right.  \left\vert \varphi
_{2}\right\rangle \right\vert , & r\leq-2\ln\left\vert \left\langle
\varphi_{1}\right.  \left\vert \varphi_{2}\right\rangle \right\vert ,\\
\infty, & r>-2\ln\left\vert \left\langle \varphi_{1}\right.  \left\vert
\varphi_{2}\right\rangle \right\vert .
\end{array}
\right.  .\nonumber
\end{align}

These confirms known results.  Interestingly, even though (\ref{C-pure-1}),
(\ref{H-pure-1}) and (\ref{C-pure-2}), (\ref{H-pure-2})  give the same
supremum, they differ at almost all the values of $\alpha.$ 

\subsection{Comparison with asymptotic optimum}

There had been a series of achievements on the asymptotic optimization of
$D_{f_{\alpha}}$. Summarizing their results (with some simple
considerations),
\[
\lim_{n\rightarrow\infty}\frac{1}{n}\log|D_{f_{\alpha}}^{\min}\left(  \rho
_{1}^{\otimes n}||\rho_{2}^{\otimes n}\right)  |=\left\{
\begin{array}
[c]{cc}%
\log\mathrm{tr}\,\left(  \rho_{2}^{\frac{1-\alpha}{2\alpha}}\rho_{1}\rho
_{2}^{\frac{1-\alpha}{2\alpha}}\right)  ^{\alpha}, & \text{ if }\alpha
\geq\frac{1}{2},\alpha\neq1,\\
\log\mathrm{tr}\,\left(  \rho_{1}^{\frac{\alpha}{2\left(  1-\alpha\right)  }%
}\rho_{2}\rho_{1}^{\frac{\alpha}{2\left(  1-\alpha\right)  }}\right)
^{1-\alpha}, & \text{if }\alpha\leq\frac{1}{2},\alpha\neq0,
\end{array}
\right.
\]
holds whenever the LHS is finite (see \ \cite{FrankLieb}%
\cite{HayashiTomamichel}\cite{MosonyiOgawa}, a brief review of them is given
in Appendix \ref{appendix:asymptotic} ). 

When $\alpha\leq1/2$ and $\rho_{1}=\left\vert \varphi_{1}\right\rangle
\left\langle \varphi_{1}\right\vert $,
\begin{align*}
\log\mathrm{tr}\,\left(  \rho_{1}^{\frac{\alpha}{2\left(  1-\alpha\right)  }%
}\rho_{2}\rho_{1}^{\frac{\alpha}{2\left(  1-\alpha\right)  }}\right)
^{1-\alpha}  & =\log\mathrm{tr}\,\left(  \left\langle \varphi_{1}\right\vert
\rho_{2}\left\vert \varphi_{1}\right\rangle \right)  ^{1-\alpha}\\
& =\log|D_{f_{\alpha}}^{\min}\left(  \rho_{1}||\rho_{2}\right)  |,
\end{align*}
so asymptotic optimal equals singe copy optimal.  

When $\alpha=2$ and $\mathrm{supp}\,\rho_{1}=\mathrm{supp}\,\rho_{2}$, the
asymptotic optimal is larger than or equal to the single copy optimal:
\begin{align*}
\log\mathrm{tr}\,\left(  \rho_{2}^{-1/2}\rho_{1}\rho_{2}^{-1/2}\right)
\rho_{1} &  =\log\sum_{i,j}\left\vert \rho_{1,i,j}\right\vert ^{2}\left(
\rho_{2,ii}\rho_{2,jj}\right)  ^{-1/2}\\
&  \geq\log\sum_{i,j}\frac{2}{\rho_{2,i,i}+\rho_{2,j,j}}\left\vert
\rho_{1,i,j}\right\vert ^{2}\\
&  =\log|D_{f_{\alpha}}^{\min}\left(  \rho_{1}||\rho_{2}\right)
|=\log\mathrm{tr}\,\left(  \frac{1}{2}T_{0}\rho_{1}\right)  .
\end{align*}
Here the identity holds only if $\rho_{2,i,i}=\rho_{2,j,j}$ for all $i$ and
$j$ with $\left\vert \rho_{1,i,j}\right\vert \neq0$, or equivallently, only if
$\rho_{2}$ commutes with $\rho_{1}$.

It is shown that an asymptotically optimal measurement sequence is projection
onto eigenspaces of $\rho_{2}^{\otimes n}$\cite{FrankLieb}%
\cite{HayashiTomamichel}\cite{MosonyiOgawa}. The key or the proof is the
inequality \
\[
A\leq v^{n}\mathcal{E}^{n}\left(  A\right)  ,
\]
where $A\geq0$ is arbitrary, $\mathcal{E}^{n}$ is the pinching operation
corresponding to the above mentioned asymptotically optimal measurement, and
$v^{n}$ is a number of distinct eigenvalues of $\rho_{2}^{\otimes n}$.

When $\alpha>1$, we can check this $\mathcal{E}^{n}$ gives asymptotically
optimal measurement using (\ref{D=sup-t}).
\begin{align*}
\frac{1}{n}\log D_{f_{\alpha}}^{\min}\left(  \rho_{1}^{\otimes n}||\rho
_{2}^{\otimes n}\right)    & =\frac{1}{n}\log\sup_{M}\sum_{t\in\mathrm{supp}%
\,M\subset\lbrack0,\infty)}\left\{  t\,\mathrm{tr}\,\rho_{1}^{\otimes n}%
M_{t}-f_{\alpha}^{\ast}\left(  t\right)  \mathrm{tr}\,\rho_{2}^{\otimes
n}M_{t}\right\}  \\
& \leq\frac{1}{n}\log\sup_{M}\sum_{t\in\mathrm{supp}\,M\subset\lbrack
0,\infty)}\left\{  t\,\mathrm{tr}\,v^{n}\mathcal{E}^{n}\left(  \rho
_{1}^{\otimes n}\right)  M_{t}-f_{\alpha}^{\ast}\left(  t\right)
\mathrm{tr}\,\rho_{2}^{\otimes n}M_{t}\right\}  \\
& =\frac{1}{n}\log\mathrm{tr}\,\left\{  v^{n}\mathcal{E}^{n}\left(  \rho
_{1}^{\otimes n}\right)  \right\}  ^{\alpha}\left(  \rho_{2}^{\otimes
n}\right)  ^{1-\alpha}\\
& =\frac{1}{n}\log\mathrm{tr}\,\left\{  \mathcal{E}^{n}\left(  \rho
_{1}^{\otimes n}\right)  \right\}  ^{\alpha}\left(  \rho_{2}^{\otimes
n}\right)  ^{1-\alpha}+\frac{\alpha}{n}\log v^{n},
\end{align*}
where the identity in the third line holds since $\mathcal{E}^{n}\left(
\rho_{1}^{\otimes n}\right)  $ commutes with $\,\rho_{2}^{\otimes n}$. On the
other hand, it is obvious
\[
\frac{1}{n}\log D_{f_{\alpha}}^{\min}\left(  \rho_{1}^{\otimes n}||\rho
_{2}^{\otimes n}\right)  \geq\frac{1}{n}\log\mathrm{tr}\,\left\{
\mathcal{E}^{n}\left(  \rho_{1}^{\otimes n}\right)  \right\}  ^{\alpha}\left(
\rho_{2}^{\otimes n}\right)  ^{1-\alpha}.
\]
Thus, since $v^{n}$ is polynomial in $n$, we have that
\[
\lim_{n\rightarrow\infty}\frac{1}{n}\log D_{f_{\alpha}}^{\min}\left(  \rho
_{1}^{\otimes n}||\rho_{2}^{\otimes n}\right)  =\lim_{n\rightarrow\infty}%
\frac{1}{n}\log\mathrm{tr}\,\left\{  \mathcal{E}^{n}\left(  \rho_{1}^{\otimes
n}\right)  \right\}  ^{\alpha}\left(  \rho_{2}^{\otimes n}\right)  ^{1-\alpha
},
\]
and the assertion is checked. 

\subsection{Kullback-Leibler divergence}

Define $f_{\mathrm{KL}}\left(  \lambda\right)  :=\lambda\ln\lambda$
($\lambda\geq0$) and $f_{\mathrm{KL}}\left(  \lambda\right)  :=\infty$
($\lambda<0$). Then
\begin{align*}
D_{f_{\mathrm{KL}}}\left(  p_{1}||p_{2}\right)   &  =D_{\hat{f}_{\mathrm{KL}}%
}\left(  p_{2}||p_{1}\right)  \\
&  =\left\{
\begin{array}
[c]{cc}%
\int p_{1}\left(  x\right)  \ln\frac{p_{1}\left(  x\right)  }{p_{2}\left(
x\right)  }\mathrm{d}\mu\left(  x\right)  , & \text{if }\mathrm{supp}%
\,p_{1}\subset\mathrm{supp}\,p_{2},\\
\infty, & \text{otherwise}.
\end{array}
\right.
\end{align*}
is the Kullback-Leibler divergence.

As easily seen,
\[
f_{\mathrm{KL}}^{\ast}\left(  t\right)  =e^{t-1}%
\]
is not operator convex but
\[
\hat{f}_{\mathrm{KL}}^{\ast}\left(  t\right)  =\left\{
\begin{array}
[c]{cc}%
-1-\ln\left(  -t\right)  , & \left(  t<0\right)  ,\\
\infty, & \left(  t\geq0\right)  ,
\end{array}
\right.
\]
is operator convex. Both of%

\begin{align*}
\mathrm{D}f_{\mathrm{KL}}^{\ast}\left(  T_{0}\right)  \left(  \rho_{2}\right)
&  =\int_{0}^{1}e^{sT_{0}}\rho_{2}e^{\left(  1-s\right)  T_{0}}\mathrm{d}%
s=\rho_{1},\\
\mathrm{D}\hat{f}_{\mathrm{KL}}^{\ast}\left(  S_{0}\right)  \left(  \rho
_{1}\right)   &  =\int_{0}^{\infty}\left(  s\mathbf{1}-S_{0}\right)  ^{-1}%
\rho_{1}\left(  s\mathbf{1}-S_{0}\right)  ^{-1}\mathrm{d}s=\rho_{2}%
\end{align*}
are difficult to solve. But using these solutions,
\begin{align*}
D_{f_{\mathrm{KL}}}^{\min}\left(  \rho_{1}||\rho_{2}\right)   &
=\mathrm{tr}\,\rho_{2}\,T_{0}^{\prime}\ln T_{0}^{\prime}\\
&  =\mathrm{tr}\,\rho_{1}\ln S_{0}^{\prime},
\end{align*}
where
\[
T_{0}^{\prime}:=e^{T_{0}-1},S_{0}^{\prime}=-S_{0}.
\]
Also, applying Theorem\thinspace\ref{th:Df-finite}, $D_{f_{\mathrm{KL}}}%
^{\min}\left(  \rho_{1}||\rho_{2}\right)  $ is finite\ only if $\mathrm{supp}%
\,\rho_{2}\supset\mathrm{supp}\,\rho_{1}$.

\subsection{Total variation distance}

Total variation distance $\left\Vert \rho_{1}-\rho_{2}\right\Vert _{1}$
equals, as is well known, $D_{f_{\mathrm{TV}}}^{\min}$, where $f_{\mathrm{TV}%
}\left(  \lambda\right)  :=\left\vert 1-\lambda\right\vert $. \
\[
f_{\mathrm{TV}}^{\ast}\left(  t\right)  =\left\{
\begin{array}
[c]{cc}%
\lambda, & \text{if }-1\leq\lambda\leq1,\\
\infty, & \text{otherwise.}%
\end{array}
\right.
\]
We confirm this result using our method. Here, it is important to choose $f$
which is canonical. Observe $f_{\mathrm{TV}}^{\ast}$ is operator convex on
$\mathrm{dom}\,f_{\mathrm{TV}}^{\ast}$ .
\begin{align*}
D_{f_{\mathrm{TV}}}^{\min}\left(  \rho_{1}||\rho_{2}\right)   &
=\sup_{T:-\mathbf{1}\leq\mathrm{spec}\,T\leq\mathbf{1}}\mathrm{tr}\,\rho
_{1}T-\mathrm{tr}\,\rho_{2}T\\
&  =\left\Vert \rho_{1}-\rho_{2}\right\Vert _{1}.
\end{align*}

Note that $t_{\ast}=1$ and $t_{\ast}^{\prime}=-1$ does not satisfy the premise
of Theorem \ref{th:T-domain}. Indeed, the supremum is achieved by a $T$ whose
eigenvalues are at the both end of the domain of $f_{\mathrm{TV}}^{\ast}$.

\section{Quantum Fisher information}

\label{sec:Fisher}

First, we present the problem rather in informal manner, before rigorous
presentation. Consider a parameterized family $\left\{  p_{\eta}\right\}
_{\eta\in\mathbb{R}}$ of probability density functions over a finite set, and
suppose $\eta\rightarrow p_{\eta}$ is smooth, and $\mathrm{supp\,}p_{\eta
}\subset\mathrm{supp\,}p_{0}$ . Suppose also $f$ is a convex function with all
good features. Then using Taylor expansion of $f$,
\begin{equation}
\lim_{\eta^{\prime}\rightarrow\eta}\frac{1}{\left(  \eta-\eta^{\prime}\right)
^{2}}\left(  D_{f}\left(  p_{\eta}||p_{\eta^{\prime}}\right)  -D_{f}\left(
p_{\eta}||p_{\eta}\right)  \right)  =\frac{1}{2}f^{\prime\prime}\left(
1\right)  J_{\eta}\label{D=J}%
\end{equation}
where
\[
J_{\eta}:=\sum_{x}\frac{\left(  \mathrm{d}p_{\eta}\left(  x\right)
/\mathrm{d}\eta\right)  ^{2}}{p_{\eta}\left(  x\right)  }%
\]
is the Fisher information of the family $\left\{  p_{\eta}\right\}  _{\eta
\in\mathbb{R}}$, which characterizes asymptotic behavior of optimal estimate
of $\eta$. The above mentioned relation is a key in the analysis of large
deviation type analysis of estimates, and in understunding its relation to
asymptotic behaviour of optimal hypothesis test.Thus exploring its quantum
analogue is also of interest. 

Let $\left\{  \rho_{\eta}\right\}  _{\eta\in\mathbb{R}}$ be a family of
density operators, and suppose $\eta\rightarrow\rho_{\eta}$ is smooth. Then
our task here is to evaluate $D_{f}^{\min}\left(  \rho_{\eta}||\rho
_{\eta^{\prime}}\right)  $ up to $O\left(  \eta-\eta^{\prime}\right)  ^{2}$.
Naively exchanging the order of limit and optimization, we have
\begin{equation}
\lim_{\eta^{\prime}\rightarrow\eta}\frac{1}{\left(  \eta-\eta^{\prime}\right)
^{2}}\left(  D_{f}^{\min}\left(  \rho_{\eta}||\rho_{\eta^{\prime}}\right)
-D_{f}^{\min}\left(  \rho_{\eta}||\rho_{\eta}\right)  \right)  =\frac{1}%
{2}f^{\prime\prime}\left(  1\right)  J_{\eta}^{S}.\label{D=JS}%
\end{equation}
where
\[
J_{\eta}^{S}:=\max_{M}J_{\eta}^{M}%
\]
and $J_{\eta}^{M}$ is the Fisher information of the family $\left\{  p_{\eta
}^{M}\right\}  _{\eta\in\mathbb{R}}$, and $p_{\eta}^{M}:=\mathrm{tr}\,$
$\rho_{\eta}M_{x}$. \ $J_{\eta}^{S}$ , \textit{SLD Fisher information}, is
given by
\[
J_{\eta}^{S}:=\mathrm{tr}\,\rho_{\eta}\left(  L_{\eta}^{S}\right)  ^{2},
\]
where $L_{\eta}^{S}$, the \textit{symmetric logarithmic derivative}
(\textit{SLD}, for short) of $\left\{  \rho_{\eta}\right\}  _{\eta
\in\mathbb{R}}$ , is defined as a Hermitian operator satisfying the equation%
\[
\frac{\mathrm{d}\rho_{\eta}}{\mathrm{d}\eta}=\frac{1}{2}\left(  L_{\eta}%
^{S}\rho_{\eta}+\rho_{\eta}L_{\eta}^{S}\right)  .
\]
SLD Fisher information, as its classical analogue, nicely characterize the
asymptotic behavior of the optimal estimate of the unknown parameter $\eta$.

If \ all members of $\{\rho_{\eta}\}_{\eta\in\mathbb{R}}$ have a common
support, it is not difficult to make the above argument rigorous. If each
$\rho_{\eta}$ differs in its support, however, the remainder term of the
Taylor expansion is not necessarily bounded due to $1/p_{\eta}^{M}$ -factors.
For example, if $\rho_{\eta}$'s ($\eta\in\mathbb{R}$) are pure states,
$D_{f_{\alpha}}^{\min}$ ($\alpha>1,\alpha<0$), $D_{f_{\mathrm{KL}}}^{\min}$
and $D_{f_{\mathrm{KL}_{2}}}^{\min}$ diverge, and (\ref{D=JS}) is never true.
On the other hand, in the case that $f^{\ast}\left(  -\infty\right)  $ is
finite, $\ f=f_{\alpha}$ ($0<\alpha\leq1/2$) for example, it is easy to see
the LHS of (\ref{D=JS}) equals constant multiple of $J_{\eta}^{S}$. Hence, the
above naive argument is not completely false, though it is not rigorous.
Below, we give deeper analysis on this issue. As it will turn out,
(\ref{D=JS}) requires some non-trivial correction, when the rank of
$\rho_{\eta}$ is neither full nor 1.

Having outlined the problem roughly, we specify all the conditions for the
argument rigorously. $f$ is supposed to be proper, convex, lower
semi-continuous, three times continuously differentiable in the neighborhood
of $1$, and
\[
f^{\prime\prime}\left(  1\right)  >0\text{.}%
\]
(All these are often assumed to derive (\ref{D=J}). Especially, the last
assumption is necessary for the $f$-divergence not to be constant at
$\eta\approx\eta^{\prime}$.) \ Also, we need additional condtions to use
Theorem\thinspace\ref{th:with-kernel}: we suppose the assumption (I) holds and
$\mathrm{dom}\,f^{\ast}$ is unbounded from below. We also suppose that the map
$\eta\rightarrow\rho_{\eta}$ is three times continuously differentiable,and
the rank of $\rho_{\eta}$ does not vary with $\eta$.

\begin{theorem}
On above mentioned conditions,
\begin{align}
&  \lim_{\eta^{\prime}\rightarrow\eta}\frac{1}{\left(  \eta^{\prime}%
-\eta\right)  ^{2}}\left(  D_{f}^{\min}\left(  \rho_{\eta}||\rho_{\eta
^{\prime}}\right)  -D_{f}^{\min}\left(  \rho_{\eta}||\rho_{\eta}\right)
\right)  \nonumber\\
&  =\frac{f^{\prime\prime}\left(  1\right)  }{2}\mathrm{tr}\,\rho_{\eta
}\left(  L_{\eta}^{S,1}\right)  ^{2}+\frac{1}{4}\left(  f^{^{\prime}}\left(
1\right)  -f\left(  1\right)  +f\left(  0\right)  \right)  \mathrm{tr}%
\,\rho_{\eta^{\prime}}\left(  L_{\eta}^{S,2}\right)  ^{2},\label{D=JS-2}%
\end{align}
where
\[
L_{\eta}^{S,1}:=\pi_{\rho_{\eta}}L_{\eta}^{S}\pi_{\rho_{\eta}},\,\,L_{\eta
}^{S,2}:=(\mathbf{1}-\pi_{\rho_{\eta}})L_{\eta}^{S}\,\pi_{\rho_{\eta}}.
\]

\end{theorem}

If $\rho_{\eta}$ is full rank ( in this case, $L_{\eta}^{S,2}=0$) or pure (in
this case $L_{\eta}^{S,1}=0$) the LHS of (\ref{D=JS-2}) equals a constant
multiple of $J_{\eta}^{S}$ , though each case differs in the value of the
constant. But if the rank of $\rho_{\eta}$ is not full nor 1, the result is a
weighted sum of two components of the SLD Fisher information; one is concerned
with the change on the support of $\rho_{\eta}$ and the other is concerned
with the change on the kernel of $\rho_{\eta}$.

For example, consider $D_{f_{\alpha}}^{\min}$ ($0<\alpha<\frac{1}{2}$). Then
\begin{align*}
&  \lim_{\eta^{\prime}\rightarrow\eta}\frac{1}{\left(  \eta^{\prime}%
-\eta\right)  ^{2}}\left(  D_{f}^{\min}\left(  \rho_{\eta}||\rho_{\eta
^{\prime}}\right)  -D_{f}^{\min}\left(  \rho_{\eta}||\rho_{\eta}\right)
\right)  \\
&  =\frac{\left(  1-\alpha\right)  \alpha}{2}\left\{  \mathrm{tr}\,\rho_{\eta
}\left(  L_{\eta}^{S,1}\right)  ^{2}+\frac{1}{2\alpha}\mathrm{tr}\,\rho_{\eta
}\left(  L_{\eta}^{S,2}\right)  ^{2}\right\}  .
\end{align*}
If $\alpha=\frac{1}{2}$, which corresponds to fidelity, this equals a constant
multiple of $J_{\eta}^{S}$. Otherwise, this gives a monotone Riemannian metric
larger than a constant multiple of $J_{\eta}^{S}$.

\begin{proof}
To prove (\ref{D=JS-2}), by Theorem\thinspace\ref{th:with-kernel}, we only
have to compute $D_{f}^{\min}\left(  \rho_{\eta}||\rho_{\eta^{\prime}%
,1}\right)  $ and $1-\mathrm{tr}\,\rho_{\eta^{\prime},1}$ up to $O\left(
\eta-\eta^{\prime}\right)  ^{2}$, where%
\[
\rho_{\eta^{\prime},1}:=\pi_{\rho_{\eta}}\rho_{\eta^{\prime}}\pi_{\rho_{\eta}}%
\]
Let $c_{\eta^{\prime}}$ be a non-positive number with
\[
\left\vert \rho_{\eta^{\prime},1}-\rho_{\eta}\right\vert -c_{\eta^{\prime}%
}\rho_{\eta}\leq0,
\]
where $\left\vert A\right\vert :=\sqrt{A^{\dagger}A}$. Suppose $\eta^{\prime}$
is enough close to $\eta$. Then $\pi_{\eta}\rho_{\eta^{\prime}}\pi_{\eta
}\approx\rho_{\eta}$, and thus $\pi_{\eta}\rho_{\eta^{\prime}}\pi_{\eta}$ is
full-rank on $\mathrm{supp}\,\rho_{\eta}$ . Therefore, $c_{\eta^{\prime}}$
exists and
\[
c_{\eta^{\prime}}\leq\frac{\left\Vert \rho_{\eta^{\prime},1}-\rho_{\eta
}\right\Vert _{1}}{r_{0}/2}=O\left(  \eta-\eta^{\prime}\right)  ,
\]
where $r_{0}$ is the smallest non-zero eigenvalue of $\rho_{\eta}$. Therefore,
for any $M_{t}\geq0$,%
\[
\frac{\left\vert \mathrm{tr}\,\left(  \rho_{\eta^{\prime},1}-\rho_{\eta
}\right)  M_{t}\right\vert }{\mathrm{tr}\,\rho_{\eta^{\prime},1}M_{t}}%
\leq\frac{\mathrm{tr}\,\left\vert \rho_{\eta^{\prime},1}-\rho_{\eta
}\right\vert \,M_{t}}{\mathrm{tr}\,\rho_{\eta^{\prime},1}M_{t}}\leq
c_{\eta^{\prime}}.
\]
Therefore, using Taylor's expansion of $f$ up to the third order, we obtain
\begin{align*}
&  D_{f}\left(  P_{\rho_{\eta}}^{M}||P_{\rho_{\eta^{\prime},1}}^{M}\right)  \\
&  =\sum_{t\in\mathrm{supp}\,M}f\left(  \frac{\mathrm{tr}\,\left(  \rho_{\eta
}-\rho_{\eta^{\prime},1}\right)  M_{t}}{\mathrm{tr}\,\rho_{\eta^{\prime}%
,1}M_{t}}+1\right)  \mathrm{tr}\,\rho_{\eta^{\prime},1}\,M_{t}\\
&  =f\left(  1\right)  \mathrm{tr}\,\rho_{\eta^{\prime},1}+f^{\prime}\left(
1\right)  \left(  1-\mathrm{tr}\,\rho_{\eta^{\prime},1}\right)  +\frac{1}%
{2}f^{\prime\prime}\left(  1\right)  \sum_{t\in\mathrm{supp}\,M}\frac{\left(
\mathrm{tr}\,\left(  \rho_{\eta}-\rho_{\eta^{\prime},1}\right)  M_{t}\right)
^{2}}{\mathrm{tr}\,\rho_{\eta^{\prime},1}M_{t}}+R,
\end{align*}
where
\begin{equation}
\left\vert R\right\vert \leq\frac{1}{6}\sup_{\lambda_{0}\in\left[
1-c_{\eta^{\prime}},1+c_{\eta^{\prime}}\right]  }\left\vert f^{\prime
\prime\prime}\left(  \lambda_{0}\right)  \right\vert c_{\eta^{\prime}}%
^{3}=O\left(  \eta^{\prime}-\eta\right)  ^{3}.\label{R<}%
\end{equation}

Since the RHS of (\ref{R<}) is independent of the measurement $M$,
\begin{align*}
&  \lim_{\eta^{\prime}\rightarrow\eta}\frac{1}{\left(  \eta-\eta^{\prime
}\right)  ^{2}}\left\{  D_{f}^{\min}\left(  \rho_{\eta}||\rho_{\eta^{\prime
},1}\right)  -D_{f}^{\min}\left(  \rho_{\eta}||\rho_{\eta}\right)  \right\}
\\
&  =\lim_{\eta^{\prime}\rightarrow\eta}\frac{1}{\left(  \eta-\eta^{\prime
}\right)  ^{2}}\left\{  D_{f}^{\min}\left(  \rho_{\eta}||\rho_{\eta^{\prime
},1}\right)  -f\left(  1\right)  \right\}  \\
&  =\frac{1}{2}f^{\prime\prime}\left(  1\right)  \lim_{\eta^{\prime
}\rightarrow\eta}\frac{1}{\left(  \eta-\eta^{\prime}\right)  ^{2}}D_{f_{b}%
}^{\min}\left(  \rho_{\eta}||\rho_{\eta^{\prime},1}\right)  \\
&  +\left(  f^{\prime}\left(  1\right)  -f\left(  1\right)  \right)
\lim_{\eta^{\prime}\rightarrow\eta}\frac{1}{\left(  \eta-\eta^{\prime}\right)
^{2}}\left(  1-\mathrm{tr}\,\rho_{\eta^{\prime},1}\right)  ,
\end{align*}
where
\[
f_{b}\left(  \lambda\right)  :=\left\{
\begin{array}
[c]{cc}%
\left(  \lambda-1\right)  ^{2}, & \text{if }\lambda\geq0,\\
-2\lambda, & \text{if }\lambda<0.
\end{array}
\right.
\]
Thus we only have to compute $D_{f_{b}}^{\min}$ and $1-\mathrm{tr}\,\rho
_{\eta^{\prime},1}$ up to $O\left(  \eta-\eta^{\prime}\right)  ^{2}$.

$D_{f_{b}}^{\min}$ is computed almost in parallel manner with $D_{f_{2}}%
^{\min}$. Since
\[
f_{b}^{\ast}\left(  t\right)  =\left\{
\begin{array}
[c]{cc}%
\frac{1}{4}t^{2}+t, & \text{if }t\geq-1/2,\\
\infty, & \text{if }t<-1/2
\end{array}
\right.
\]
is operator convex on $\mathrm{dom}\,f_{b}^{\ast}$, by (\ref{rho=Df}), using a
solution $T_{0}$ to the follwoing Lyapunov equation,
\begin{equation}
\rho_{\eta}-\rho_{\eta^{\prime},1}=\frac{1}{4}\left(  T_{0}\rho_{\eta^{\prime
},1}+\rho_{\eta^{\prime},1}T_{0}\right)  ,\label{T-rho}%
\end{equation}
we have%
\begin{equation}
D_{f_{b}}^{\min}\left(  \rho_{\eta}||\rho_{\eta^{\prime},1}\right)
=\mathrm{tr}\,\left(  \frac{1}{2}T_{0}+\mathbf{1}-\mathbf{1}\right)  ^{2}%
\rho_{\eta^{\prime},1}=\mathrm{tr}\,\left(  \frac{1}{2}T_{0}\right)  ^{2}%
\rho_{\eta^{\prime},1}.\label{Dfb}%
\end{equation}

Here, that (\ref{T-rho}) has a solution is checked by Theorem\thinspace
\ref{th:T-domain} and the fact that $\mathrm{dom}\,f_{b}^{\ast}=[-\frac{1}%
{2},\infty)$. Alternatively, the solution $T_{0}$ can be explicitly
constructed almost parallel manner as in the case of  $f_{2}$.

Since its rank does not vary with $\eta$, $\rho_{\eta}$ can be written as
$\rho_{\eta}=A_{\eta}A_{\eta}^{\dagger}$, and the curve $\left\{  A_{\eta
}\right\}  _{\eta\in\mathbb{R}}$ satisfies
\begin{equation}
A_{\eta^{\prime}}=A_{\eta}+\frac{1}{2}\left(  \eta^{\prime}-\eta\right)
L_{\eta}^{S}A_{\eta}+C_{1}\label{A-def}%
\end{equation}
for some $C_{1}=O\left(  \eta^{\prime}-\eta\right)  ^{2}$. By (\ref{T-rho})
and (\ref{A-def}), we have
\[
\frac{1}{2}T_{0}=-\left(  \eta^{\prime}-\eta\right)  L_{\eta}^{S,1}+C_{2},
\]
where $C_{2}$ is $O\left(  \eta^{\prime}-\eta\right)  ^{2}$. Inserting this
into (\ref{Dfb}), $D_{f_{b}}^{\min}\left(  \rho_{\eta}||\rho_{\eta^{\prime}%
,1}\right)  $ is computed up to $O\left(  \eta^{\prime}-\eta\right)  ^{2}$.

On the other hand, $1-\mathrm{tr}\,\rho_{\eta^{\prime},1}$ is computed as
follows.
\begin{align*}
1-\mathrm{tr}\,\rho_{\eta^{\prime},1} &  =\mathrm{tr}\,(\mathbf{1}-\pi
_{\rho_{\eta}})\rho_{\eta^{\prime}}(\mathbf{1}-\pi_{\rho_{\eta}})\\
&  =\mathrm{tr}\,(\mathbf{1}-\pi_{\rho_{\eta}})\left(  \frac{1}{2}L_{\eta}%
^{S}A_{\eta}\left(  \eta^{\prime}-\eta\right)  +C_{1}\right)  \left(  \frac
{1}{2}L_{\eta}^{S}A_{\eta}\left(  \eta^{\prime}-\eta\right)  +C_{1}\right)
^{\dagger}(\mathbf{1}-\pi_{\rho_{\eta}})\\
&  =\mathrm{tr}\,\left(  \frac{1}{2}L_{\eta}^{S,2}A_{\eta}\left(  \eta
^{\prime}-\eta\right)  +(\mathbf{1}-\pi_{\rho_{\eta}})C_{1}\right)  \left(
\frac{1}{2}L_{\eta}^{S,2}A_{\eta}\left(  \eta^{\prime}-\eta\right)
+(\mathbf{1}-\pi_{\rho_{\eta}})C_{1}\right)  ^{\dagger}\\
&  =\frac{1}{4}\left(  \eta^{\prime}-\eta\right)  ^{2}\mathrm{tr}\,L_{\eta
}^{S,2}\rho_{\eta}L_{\eta}^{S,2}+O\left(  \eta^{\prime}-\eta\right)  ^{3}.
\end{align*}
After all, we have (\ref{D=JS-2}).
\end{proof}

\section{Summary and questions}

Using tools from convex analysis and matrix analysis, the maximization of the
measured $f$-divergence is rewritten to a simpler form (\ref{Df-sup-T}) on the
assumption that (I) or (II) holds, and derived and proved some closed formulas
and properties of $D_{f}^{\min}$. Some questions are in order. First, what is
the necessary and sufficient condition of $f$ such that (\ref{Df-sup-T})
holds? Second, the condition (I) is written in terms of $f^{\ast}$, but it
would be nicer to have some alternative condition written in terms of $f$
\ itself, since in the study of other versions of quantum $f$-divergence, they
often assume $f$ to be operator convex (and not $f^{\ast}$). Third, it is easy
to obtain a lower bound to $D_{f}^{\min}$ by (\ref{Df-sup-T}), but is there
any good upper bound to the quantity, which is useful for the study of
asymptotic theory?

\appendix

\section{The proof of (\ref{f_0-1}), (\ref{f_0-2}), and (\ref{t0=f'})}

\label{sec:proof-f_0}

To see that $f_{0}$ and $f$ coincide on the positive half-line, observe that
$f^{\ast}$ is monotone non-increasing in the region below $\mathrm{dom}%
\,f_{0}^{\ast}$. Therefore, in that region, $t\rightarrow t\lambda-f^{\ast
}\left(  t\right)  $ is monotone non-decreasing if $\lambda\geq0$. Thus,
\begin{align*}
f\left(  \lambda\right)   &  =\sup_{t\in\mathrm{dom}\,f_{0}^{\ast}}%
t\lambda-f^{\ast}\left(  t\right)  \\
&  =\sup_{t\in\mathrm{dom}\,f_{0}^{\ast}}t\lambda-f_{0}^{\ast}\left(
t\right)  \\
&  =f_{0}\left(  \lambda\right)  .
\end{align*}

Suppose $\lambda<0$. Then in taking supremum of $t\lambda-f_{0}^{\ast}\left(
t\right)  $, the range of $t$ can be limitted to $\mathrm{dom}\,f_{0}^{\ast}$.
Also, $t\lambda-f_{0}^{\ast}\left(  t\right)  $ is monotone decreasing on
$\mathrm{dom}\,f_{0}^{\ast}$ . Therefore,
\begin{align*}
f_{0}\left(  \lambda\right)    & =\sup_{t\in\mathrm{dom}\,f_{0}^{\ast}%
}t\lambda-f_{0}^{\ast}\left(  t\right)  \\
& =t_{0}\lambda-f_{0}^{\ast}\left(  t_{0}\right)  \\
& =t_{0}\lambda+f\left(  0\right)  ,
\end{align*}
where the third equality is by (\ref{t0=inf}).

Thus it remains to show (\ref{t0=f'}) in the $t_{0}>-\infty$-\thinspace case.
To show this, suppose the contrary is true. Since $f_{0}$ is convex,
$t_{0}<f_{+}^{\prime}\left(  0\right)  $. Here note $f_{+}^{\prime}\left(
0\right)  $ is finite since and $f$ is finite at some point on the positive
half line. \ Let $t$ be a real number between $t_{0}$ and $f_{+}^{\prime
}\left(  0\right)  $. Then%
\begin{align*}
f_{0}^{\ast}\left(  t\right)   &  =\max\left\{  \sup_{\lambda>0}\left(
t\lambda-f\left(  \lambda\right)  \right)  ,\sup_{\lambda\leq0}\left(
t\lambda-\left(  \lambda t_{0}+f\left(  0\right)  \right)  \right)  \right\}
\\
&  \leq\max\left\{  \sup_{\lambda>0}\left(  \left(  t-f_{+}^{\prime}\left(
0\right)  \right)  \lambda-f\left(  0\right)  \right)  ,\sup_{\lambda\leq
0}\left(  \left(  t-t_{0}\right)  \lambda-f\left(  0\right)  \right)
\right\}  \\
&  =-f\left(  0\right)  =f^{\ast}\left(  t_{0}\right)  =f_{0}^{\ast}\left(
t_{0}\right)  .
\end{align*}
This contradicts with the requirement that $f_{0}^{\ast}$ should be strictly
monotone increasing in its effective domain. Thus we should have
(\ref{t0=f'}). 

\section{The proof of Theorem\thinspace\ref{th:DfQ-finite}}

\label{sec:proof-th-finite}

By (\ref{Dfmin<DfQ}) and Theorem\thinspace\ref{th:Df-finite}, obviously these
conditions are necessary for $D_{f}^{Q}$ to be finite. Thus, we show the
conditions are sufficient.

Suppose $\mathrm{supp}\rho_{1}\subset\mathrm{supp}\,\rho_{2}$ . Then there is
a number $\lambda\in\left[  0,1\right]  $ $\ $and $\rho_{3}$ with $\rho
_{2}=\lambda\rho_{1}+\left(  1-\lambda\right)  \rho_{3}$. Define $P_{\theta}$
on the set $\left\{  +,-\right\}  $ as follows:%
\begin{align*}
P_{1}\left(  +\right)    & :=1,P_{1}\left(  -\right)  :=0,\\
P_{2}\left(  +\right)    & :=\lambda,P_{2}\left(  -\right)  :=1-\lambda.
\end{align*}
Also, define a CPTP map $\Lambda$ by $\Lambda\left(  \delta_{+}\right)
:=\rho_{1}$ and $\Lambda\left(  \delta_{-}\right)  :=\rho_{3}$. Then
$\Lambda\left(  P_{\theta}\right)  =\rho_{\theta}$ and $\mathrm{supp}%
P_{1}\subset\mathrm{supp}\,P_{2}$. Thus if (\ref{Wf-bounded-1}) holds,
\[
D_{f}^{Q}\left(  \rho_{1}||\rho_{2}\right)  \leq D_{f}\left(  P_{1}%
||P_{2}\right)  <\infty.
\]
Almost parallely, If $\mathrm{supp}\rho_{2}\subset\mathrm{supp}\,\rho_{1}$,
(\ref{Wf-bounded-2}) implies $D_{f}^{\max}\left(  \rho_{1}||\rho_{2}\right)
<\infty$. 

Finally, if  (\ref{Wf-bounded}) holds, for any $P_{\theta}$, $D_{f}\left(
P_{1}||P_{2}\right)  <\infty$. \ Let $P_{1}:=\delta_{+}$ and $P_{2}%
:=\delta_{-}$ . Then,  Also, $\Lambda\left(  P_{\theta}\right)  =\rho_{\theta
}$ holds with  $\Lambda\left(  \delta_{+}\right)  :=\rho_{1}$ and
$\Lambda\left(  \delta_{-}\right)  :=\rho_{2}$. Therefore,  $D_{f}^{Q}\left(
\rho_{1}||\rho_{2}\right)  \leq D_{f}\left(  P_{1}||P_{2}\right)  <\infty$,
and the proof is complete.

\section{The proof of Lemma\thinspace\ref{lem:monotone}}

\label{sec:proof-lem-monotone}

First we show $\mathrm{dom}\,f^{\ast}=\left(  -\infty,a\right)  $ or
$(-\infty,a]$. Suppose the contrary, or equivallently, $\mathrm{dom}\,f^{\ast
}$ is the whole real line. Since $f^{\ast}$ is operator convex, it is
quadratic. Then, since $f$ is canonical, or equivallently $f^{\ast}$ is
monotone increasing, $f^{\ast}\left(  -\infty\right)  =-\infty$,
contradicting  $f^{\ast}\left(  -\infty\right)  >-\infty$. Thus the assertion
holds.  

Suppose $\mathrm{dom}\,f^{\ast}=(-\infty,a]$. Then, $h\left(  t\right)
:=f^{\ast}\left(  -t+a\right)  $ is finite on $[0,\infty)$ and monotone
non-increasing. Also,  $h\left(  \infty\right)  >-\infty$. \ Since $h$ is
monotone non-increasing and proper, $h\left(  \infty\right)  <\infty$, and
\[
\lim_{t\rightarrow\infty}\frac{h\left(  t\right)  }{t}=0.
\]

Therefore, by Proposition\thinspace8.4 of \cite{HiaiMosonyiPetzBeny}, $h$ can
be written as
\[
h\left(  t\right)  =h\left(  0\right)  +\alpha t-\int_{\left(  0,\infty
\right)  }\frac{t}{t+\eta}\mathrm{d}\nu\left(  \eta\right)
\]
using a non-negative measure $\nu$ with $\int_{\left(  0,\infty\right)  }%
\frac{1}{1+\eta}\mathrm{d}\nu\left(  \eta\right)  <\infty.$\ Since
\[
\alpha=\frac{h\left(  t\right)  -h\left(  0\right)  }{t}+\int_{\left(
0,\infty\right)  }\frac{1}{t+\eta}\mathrm{d}\nu\left(  \eta\right)
\]
and
\[
\lim_{t\rightarrow\infty}\int_{\left(  0,\infty\right)  }\frac{1}{t+\eta
}\mathrm{d}\nu\left(  \eta\right)  =0
\]
by Lebesgue's dominated convergence theorem, we have $\alpha=0$. Since for
each $\eta$ the function $t\rightarrow-t/\left(  t+\eta\right)  $ is oprator
monotone decreasing, $h\left(  t\right)  $ is operator monotone decreasing,
implying the assertion.

If $\mathrm{dom}\,f^{\ast}=(-\infty,a)$, due to the above argument, $f^{\ast}$
is operator monotone increasing on $(-\infty,a-\varepsilon]$ for any
$\varepsilon>0$. Suppose the spectrum of $A_{1}$ and $A_{2}$ is a subset of
$(-\infty,a]$, and $A_{1}\geq A_{2}$. Then since $A_{1}-\varepsilon
\mathbf{1}\geq A_{2}-\varepsilon\mathbf{1}$,
\[
f^{\ast}\left(  A_{1}-\varepsilon\mathbf{1}\right)  \geq f^{\ast}\left(
A_{2}-\varepsilon\mathbf{1}\right)  .
\]
Letting $\varepsilon\rightarrow0$, we have $f^{\ast}\left(  A_{1}\right)  \geq
f^{\ast}\left(  A_{2}\right)  $, meaning that $f^{\ast}$ is operator monotone
increasing on $(-\infty,a)$.

\section{Maximization of asymptotic measured relative Renyi entropy}

\label{appendix:asymptotic}

This appendix is a brief review of  \cite{MosonyiOgawa}  and
\cite{HayashiTomamichel}. When $\mathrm{supp}\,\rho_{1}=\mathrm{supp}%
\,\rho_{2}$, we define
\begin{equation}
\tilde{D}_{f_{\alpha}}\left(  \rho_{1}||\rho_{2}\right)  :=\mathrm{sign}%
\left(  \left(  1-\alpha\right)  \alpha\right)  \cdot\mathrm{tr}\,\left(
\rho_{2}^{\frac{1-\alpha}{2\alpha}}\rho_{1}\rho_{2}^{\frac{1-\alpha}{2\alpha}%
}\right)  ^{\alpha}.\label{tildeD-1}%
\end{equation}
$\tilde{D}_{f_{\alpha}}$ is extended to the case of $\mathrm{supp}\,\rho
_{1}\neq\mathrm{supp}\,\rho_{2}$ so that the function is lower semicontinuous.
Namely, we define
\begin{equation}
\tilde{D}_{f_{\alpha}}\left(  \rho_{1}||\rho_{2}\right)  :=\infty
,\label{tildeD-2}%
\end{equation}
for some value of $\alpha$ : we adopt (\ref{tildeD-2}) for the interval
$\alpha\in(1,\infty)$ if $\ker\rho_{2}\cap\mathrm{supp}\,\rho_{1}\neq\left\{
0\right\}  $, and for the interval $\alpha\in(-\infty,0)$, \ if $\ker\rho
_{1}\cap\mathrm{supp}\,\rho_{2}\neq\left\{  0\right\}  $: for the interval
$\alpha\in(0,1)$, we always adopt (\ref{tildeD-1}).

Suppose $\alpha\in\lbrack\frac{1}{2},1)\cup(1,\infty)$. In this case, by
\cite{FrankLieb}, $\tilde{D}_{f_{\alpha}}$ is monotone non-increasing by
CPTP\ maps, or
\[
\tilde{D}_{f_{\alpha}}\left(  \rho_{1}||\rho_{2}\right)  \geq\tilde
{D}_{f_{\alpha}}\left(  \Lambda\left(  \rho_{1}\right)  ||\Lambda\left(
\rho_{2}\right)  \right)
\]
for any CPTP map $\Lambda$. Also, obviously additivity
\[
\log\left\vert \tilde{D}_{f_{\alpha}}\left(  \rho_{1}^{\otimes n}||\rho
_{2}^{\otimes n}\right)  \right\vert =n\log\left\vert \tilde{D}_{f_{\alpha}%
}\left(  \rho_{1}||\rho_{2}\right)  \right\vert
\]
holds and restriction to commutative states equals $D_{f_{\alpha}}$. Thus,
$\tilde{D}_{f_{\alpha}}$ is an upper bound to asymptotic optimal,
\[
\lim_{n\rightarrow\infty}\frac{1}{n}\log\left\vert D_{f_{\alpha}}^{\min
}\left(  \rho_{1}^{\otimes n}||\rho_{2}^{\otimes n}\right)  \right\vert
\leq\log\left\vert \tilde{D}_{f_{\alpha}}\left(  \rho_{1}||\rho_{2}\right)
\right\vert .
\]
The achievability part is shown by \cite{MosonyiOgawa} ($\alpha>1$) and
\cite{HayashiTomamichel} ($\alpha\in(0,1)\cup(1,\infty)$ ).

When $\alpha\in(-\infty,0)\cup(0,\frac{1}{2}]$, we only have to exploit the
fact $D_{f_{\alpha}}\left(  P_{1}||P_{2}\right)  =D_{f_{1-\alpha}}\left(
P_{2}||P_{1}\right)  $.

\begin{remark}
In Remark III.5 of \cite{MosonyiOgawa}, they state that the monotonicity of
$\tilde{D}_{f_{\alpha}}$ ($\alpha>1$) by CPTP maps is not true for the states
with $\ker\rho_{2}\cap\mathrm{supp}\,\rho_{1}\neq\left\{  0\right\}  $. We had
avoided this problem by using the definition (\ref{tildeD-2}).
\end{remark}

\section{}
\end{document}